\documentclass[a4paper,UKenglish,cleveref, autoref]{lipics-v2021} 

\usepackage{Luckhardt_layout_LIPIcs}
\usepackage{Luckhardt_math}
\usepackage{custom}

\title{Expressivity of bisimulation pseudometrics over analytic state spaces}

\author{Daniel Luckhardt}{UCL, University College London, United Kingdom}{d.luckhardt@sheffield.ac.uk}{0000-0002-1886-5622}{}
\author{Harsh Beohar}{University of Sheffield, Sheffield, United Kingdom}{h.beohar@sheffield.ac.uk}{0000-0001-5256-1334}{}

\author{Clemens Kupke}{University of Strathclyde, Glasgow, United Kingdom}{clemens.kupke@strath.ac.uk}{0000-0002-0502-391X}{}

\authorrunning{D. Luckhardt, H. Beohar and C. Kupke}

\Copyright{Daniel Luckhardt, Harsh Beohar and Clemens Kupke} 



\category{} 

\relatedversion{} 

\EventNoEds{2}
\EventLongTitle{11th Conference on Algebra and Coalgebra in Computer Science (CALCO25)}
\EventShortTitle{CALCO25}
\EventAcronym{CVIT}
\EventYear{2016}
\EventDate{December 24--27, 2016}
\EventLocation{Little Whinging, United Kingdom}
\EventLogo{}
\SeriesVolume{42}
\ArticleNo{23}

\excludecomment{CALCOversion}
\includecomment{arXivVersion}

\presetkeys{todonotes}{disable}{}

\begin{document}

\nolinenumbers

\maketitle

\begin{abstract}
    A Markov decision process (MDP) is a state-based dynamical system capable of describing probabilistic behaviour with rewards. In this paper, we view MDPs as coalgebras living in the category of analytic spaces, a very general class of measurable spaces. Note that analytic spaces were already studied in the literature on labelled Markov processes and bisimulation relations. \cktodo{some people might think we claim we are the fist to do view MDPs as coalg over Ana?} Our results are twofold. First, we define bisimulation pseudometrics over such coalgebras using the framework of fibrations. Second, we develop a quantitative modal logic for such coalgebras and prove a quantitative form of Hennessy-Milner theorem in this new setting stating that the bisimulation pseudometric corresponds to the logical distance induced by modal formulae.
\end{abstract}

\section{Introduction}
    \label{sec:intro}
Markov decision processes (MDPs) are a well known mathematical model for decision-theoretic planning \cite{DTP-MDP} and reinforcement learning \cite{RL-MDP}. Informally, an MDP can be seen as a generalisation of an automaton, where the transition function (for each action in the alphabet) gives a probability distribution over the state space together with a reward function that for each state and action gives a real-valued number.

Inspired from the previous work on bisimulation pseudometrics on labelled Markov processes \cite{DesharnaisEt-Concur99,DesharnaisEt2004:LMP} and probabilistic transition systems \cite{vBW-ICALP01,vBW-CONCUR01}, Ferns et al.\nolinebreak[3]\ \cite{FernsPanangadenPrecup04,FernsPanangadenPrecup11} defined (among other results) a notion of bisimulation pseudometric on the states of an MDP. Unlike \cite{FernsPanangadenPrecup04} and the previous work on bisimulation equivalence for MDPs \cite{GIVAN2003:MDP}, the systems considered in \cite{FernsPanangadenPrecup11} were MDPs with continuous state spaces.
Conformances over continuous state MDPs have found applications in representation learning \cite{pmlr-v97-gelada19a,zhang2021learning}  (a topic studied within the field of reinforcement learning).


In this paper, we propose a modal logic $\Sprache$ (cf.\nolinebreak[3]\ \cref{sec:log-MDPs}) with quantitative semantics, i.e. the semantics of each formula is given by a real-valued function for MDPs with continuous state space. We then prove a quantitative version of the Hennessy-Milner theorem (a well known result \cite{HML-1980} from concurrency theory), i.e. we show that the bisimulation pseudometrics on continuous state MDPs coincide with the logical distance in our logic. 
A major obstacle to overcome in the continuous setting is the definition of bisimulation pseudometrics itself. Moreover, the fundamental question ``what is a distance on a measurable space'' (besides the usual equations of a pseudometric, i.e.\nolinebreak[3]\ $d$ is reflexive, symmetric, and satisfies the triangle inequality) needs addressing. In \cite{FernsPanangadenPrecup11} the authors had to invoke an additional Polish structure inducing the \salg{} as their methods forced them to work with lower semi-continuous distance functions $d$. In this sequel, we work in a purely measure theoretic set-up with a far more general class of distances, universally measurable distance functions, cf.\ \cref{subsec:pred} for details.  \cktodo{what does general mean? more general?}\dltodo{I wouldn't put lower semi-measurability in the foreground, but universal measurability. The latter is a known concepts, while the former more auxiliary in this paper (perhaps we could say in future works that this concept may deserve more attention).}

Although our approach is rooted in the theory of fibrations \cite{JacobsCLTT}, the recent approaches \cite{beohar_et_al:LIPIcs.STACS.2024.10,KomoridaKatsumataKu-ExpressivityofQuant,KupkeRot-CoindPred-2020} to obtain expressive modal logic for coalgebras do not apply. For instance, our fibration $\predicates$ of predicates\footnote{Note that a fibration of predicates is more fundamental than a fibration of conformances like pseudometrics and equivalence relations, since the latter can be derived from the former.} over a state space has only countable many meets; thus, it is not a complete lattice fibration as required in \cite{beohar_et_al:LIPIcs.STACS.2024.10,KomoridaKatsumataKu-ExpressivityofQuant,KupkeRot-CoindPred-2020}. As a result, the codensity lifting used in \cite{beohar_et_al:LIPIcs.STACS.2024.10,KomoridaKatsumataKu-ExpressivityofQuant} to derive the Kantorovich lifting for the distribution endofunctor, cannot be used to derive the Kantorovich lifting for the Giry endofunctor over measurable spaces.

In the sequel, we recalibrate the fibration infrastructure in \cref{ssec:catnonsense}. Our inspiration is \cite{BonchiKonigPetrisan18} which presented a coupling-based lifting for an endofunctor that---when instantiated to distribution endofunctor---gives rise to the well-known Wasserstein lifting on probability distributions. Analogously, we will show in \cref{subsec:wlift} how to capture the Wasserstein lifting on probability measures. For our definition to work, we will restrict to a full subcategory of the category $\measSp$ of measurable spaces -- the category $\ana$ of analytic spaces.
Note that analytic spaces already appeared in the literature on labelled Markov processes (for instance, see \cite{DESHARNAIS2002163}) to show that logical equivalence induced by a modal logic given in \cite{DESHARNAIS2002163} coincide with probabilistic bisimilarity.

After having clarified our measure theoretic assumptions, 
we will define bisimulation metrics for MDPs as the least fixpoint of the following functional:
\[
\predicates (\measSp \times \measSp)
\xrightarrow{\sigma_X}
\predicates (\BfunctorMDP\measSp \times \BfunctorMDP\measSp)
\xrightarrow{\smash{(\gamma\times \gamma)^*}}
\predicates (\measSp \times \measSp),
\]
where $\BfunctorMDP$  is the endofunctor modelling MDPs as given in \cref{sec:MDP} and $\sigma$ is a lifting of distance functions (or put simply, a distance lifting)\cktodo{what is "the distance lifiting" referring to - completely unclear}\dltodo{Does "Wasserstein" improves the situation?}\hbtodo{No, it doesn't} 
for $\BfunctorMDP$ as given in \cref{sec:bdistance}.
The definition of our distance lifting $\sigma$ is parameterised by a discount factor $c\in[0,1]$.  Furthermore, thanks to the Kantorovich-Rubinstein duality for measurable spaces \cite[Theorem~5]{RamachandranRüschendorf95}, the above functional corresponds to the functional given in \cite[Theorem~3.12]{FernsPanangadenPrecup11} whose least fixpoint is the bisimulation pseudometric on the state space of an MDP.

Moving on to our modal logic and comparing with the expressive modal logic for probabilistic systems studied in \cite{ChenClercPanaganden25,DesharnaisEt2004:LMP,vBW-ICALP01}, the key distinguishing feature of our work\cktodo{please check} is the semantics of our diamond modality $\diamond_a \varphi$. Intuitively, the $\interpret{\diamond_a \varphi}(x)$ (for a state $x$) gives the expected value of landing in an $a$-successor from $x$ with some fixed probability $c \in [0,1]$ combined with the reward for $a$ when staying in the state $x$ with probability $1-c$. In other words, the semantics of $\diamond_a \varphi$ is a convex combination of expected value of moving to an $a$-successor and the reward for $a$ at a state. Unlike the above references, we were unable (without breaking the proof of the adequacy result) to further decompose this modality into the traditional diamond modality and 0-ary reward modality as defined in \cite{ChenClercPanaganden25,DesharnaisEt2004:LMP,vBW-ICALP01}.

This paper is organised as follows. In \cref{sec:prelim}, we recall the preliminaries from measure theory and calibrate our fibration setup for measurable spaces. In \cref{sec:MDP}, we give the concrete definition of behaviour endofunctors that model Markov reward processes (MRPs) and MDPs and establish a bifibration of predicates. The former can be seen as unlabelled version of an MDP. In \cref{sec:bdistance}, we capture the bisimulation pseudometrics for both MRPs and MDPs as least fixpoint of a functional as explained above. In \cref{sec:log-MDPs}, we define our modal logic and establish the adequacy and expressivity results. In \cref{sec:conc}, we end this paper by a discussion on related work and potential topics for future work. The proofs of all lemmas and theorems can be found in the clearly marked appendix.

\cktodo{Main "weakness" for me currently is that it's not clear how we improve the existing fibration-based approaches - is this "just" by weakening the order-theoretic assumptions? In what sense are we here inspired by Petrisan et al?}

\section{Preliminaries}\label{sec:prelim}

\subsection{Capturing behavioural conformances categorically}
\label{ssec:catnonsense}
%
In this subsection, we refine the construction \cite{BonchiKonigPetrisan18} of coupling-based lifting for an endofunctor on $\set$ by working with two different fibrations of predicates (cf.\nolinebreak[3]\ Assumptions~\cref{ass:a1} and \cref{ass:a2}). 
Moreover, our presentation works in a category $\cat$ having products; unlike, in \cite{BonchiKonigPetrisan18}, where the coalgebras were living in $\set$. 
This will provide us a blueprint to define a bisimulation distance for both MRPs and MDPs when viewed as coalgebras in \cref{sec:MDP}. 

Throughout this section, let $\pos$ be the category of posets and order preserving maps; and, let $B\colon \cat \to \cat$ be the functor modelling the branching type of systems of interest.

\begin{axiom}
    \item\label{ass:a1} There is an indexed category ${\predicates}\colon \cat^{\text{op}} \to \pos$ such that $\pred X$ (for $X\in\cat$) is a poset and $\pred f\colon \pred Y \to \pred X$ (for $f\colon X\to Y \in \cat$) is order preserving.
\end{axiom}
Henceforth we write the reindexing $f^*$ (instead of $\pred f$) which is customary in the literature on fibrations \cite{JacobsCLTT}. In the sequel, we will view an element $p\in \pred X$ intuitively as a predicate over an object $X\in\cat$. The idea is to view $\predicates$ as a semantic universe in which we interpret the formulae of a modal logic. Thus, the operators of a modal logic (like negation, conjunction etcetera) must be operators definable over the fibre $\pred X$. 

Furthermore, the authors in \cite{BonchiKonigPetrisan18}  required that $\predicates$ is rather a bifibration, which is difficult to obtain in general for arbitrary measurable spaces (cf.\nolinebreak[3]\ \cref{subsec:pred}). Our observation, which leads to \cref{ass:a2}, is that we can arrange both universally measurable predicates and lower semi-measurable predicates in such a way that the latter results in a bifibration structure and the former acts as a semantic universe to interpret our modal formulae.
\begin{axiom}
\setcounter{axiomctr}{1}
    \item\label{ass:a2} there is an indexed category $\spredicates$ such that $\spredicates$ is a subfunctor of $\predicates$. 
    Moreover, the indexed category $\spredicates$ has a bifibration structure, i.e.\nolinebreak[3]\ for every $f\colon X\to Y\in \cat$ the reindexing functor $f^*$ has a left adjoint $\exists_f\colon\spred X \to \spred Y$.
\end{axiom}

Now, following \cite{BonchiKonigPetrisan18}, one needs a predicate lifting to define a coupling-based lifting, which in our setting due to the presence of two fibrations of predicates takes the following shape.
\begin{axiom}
\setcounter{axiomctr}{2}
    \item\label{ass:a3} there is an indexed morphism $\sigma \colon {\predicates} \Rightarrow \spredicates \circ B^\text{op}$, i.e. $\sigma$ is a natural transformation of type ${\predicates} \Rightarrow \spredicates \circ B^\text{op}$.
\end{axiom}
Thanks to the above three assumptions, every predicate lifting $\sigma$ induces a lifting $\hat\sigma$, which simplifies to the composition given in \cite[Eq.~5]{BonchiKonigPetrisan18} when $\spredicates =\predicates$). 
\begin{equation}\label{eq:clift}
\pred {X\times X} \xrightarrow{\sigma_{X\times X}} \spred {B(X\times X)} 
\xrightarrow{\exists_{\pi_X}} \spred {BX \times BX} \hookrightarrow \pred {BX \times BX},
\end{equation}
where $\pi_X \colon B(X\times X) \to BX \times BX$ is the unique map such that $\proj i^{BX} \circ \pi_X = B(\proj i^X)$ and $\proj i \colon X \times X \to X $ are the obvious projection maps (for $i\in\{1,2\}$).
It is this lifting which will give us the usual Wasserstein lifting for a Giry functor $B$ defined in \cref{sec:MDP}. 
Now, for a given coalgebra $\gamma\colon X \to BX \in \cat$, simply take the greatest fixpoint of the functional given below to define a \emph{coupling-based lifting} for the endofunctor $B$.
\begin{equation}\label{eq:gfp-clift}
\pred{X \times X} \xrightarrow{\hat \sigma_X} \pred{ BX \times BX } \xrightarrow{(\gamma \times\gamma)^*} \pred{X \times X}.
\end{equation}
To ensure this fixpoint exists, we require the following assumption
\begin{axiom}
\setcounter{axiomctr}{3}
    \item\label{ass:a5} the indexed category $\predicates$ has countable fibred limits, i.e.\nolinebreak[3]\ each fibre of $\predicates$ has countable meets and these countable meets are preserved by the reindexing operation.
\end{axiom}
\begin{proposition}\label{prop:KFixpoint}
    If the induced lifting $\hat\sigma$ defined in \eqref{eq:clift} is (Scott) cocontiuous, then the greatest fixpoint of the functional given in \eqref{eq:gfp-clift} exists.
\end{proposition}
\begin{remark}
It should be noted that, in the above proposition, we use the dual version to apply Kleene's fixpoint theorem on the lattice $\predicates$. However, our concrete predicates in Section~\ref{sec:MDP} will be ordered by pointwise lifting of the dual order (i.e.\nolinebreak[3]\ $\geq$) on the unit interval $[0,1]$, which then leads to the application of the usual Kleene's fixpoint theorem. 
\end{remark}

\subsection{Measurable spaces}
    \label{ssec:measurableSp}
A measurable space is a pair $ \measSp $ consisting of a set $\spSet$ thought of as a space, e.g.\nolinebreak[3]\ state space of an MDP, and a $\sigma$-algebra $\sAlg \subseteq \powerSet(\spSet)$ (whose elements are called \definiendum{measurable} sets) of subsets of $\spSet$ stable under the complement operation $\complOp{\blank}$ and countable union (including empty union). 
In applications the \salg{} often contains a given topology, i.e.\ the collection of open sets $\topy$, on the state space.
Often one considers the minimal such \salg{}, the \definiendum{Borel-\salg}, denoted $\Borel{\topy}$.
In this context, we use the notation $\pavGen\sigma{\paving P} $ to denote the minimal \salg{} containing a given $\paving P \subseteq \powerSet(\spSet)$.
The elements of $\Borel{\topy}=\pavGen\sigma{\topy}$ are called \emph{Borel} sets and $ \BorelSp{\topy} $ is called a \definiendum{Borel space}.
If a given measurable space $\measSp$ stems from a Polish space, a completely metrisable\cktodo{reference on basics on metric spaces?} separable topological space, then $\measSp$ is called \definiendum{standard}. 
The collection of all measurable spaces form a category $\Cat{Meas}$ when endowed with maps that inversely preserve measurable sets, so-called \definiendum{measurable maps}.
The term \definiendum{$\sAlg$-$\sAlg*$-measurable} for a measurable morphism $ \measSp \to \measSp* $ is also used to be specific about \salg{}s.
The category $\Meas$ has arbitrary products; in particular, $\measSp \times \measSp* = (\spSet\times \spSet*, \sAlg \measTimes \sAlg*)$ where $\sAlg\otimes \sAlg*$ is the $\sigma$-algebra generated by the set $\setBuilder{U \times V}{ U \in \sAlg \land V \in\sAlg*}$.

A \emph{probability measure} $\meas$ on a measurable space $\measSp$ is a function of type $\sAlg \to [0,1]$ such that $\meas(X) = 1$ and $\meas( \bigcup_{i\in\nat} A_i ) = \sum_{i\in \nat} \meas(A_i)$ for any sequence of pairwise disjoint sets $(A_i\in\sAlg)_{i\in\nat}$. 
We denote by $\giry\measSp$ the collection of all probability measures on $\measSp$ endowed---going back to \cite{Giry82}---with the \salg{} generated by all the evaluation maps $\eval{A}\colon \giry \measSp \to [0,1]$ (one for each $A \in \sAlg$) given by the mapping $\meas \mapsto \meas(A)$, i.e.\ the minimal \salg{} making all maps $ \eval{A}\colon \giry \measSp \to \BorelSp{[0,1]} $ measurable. 


We restrict the exposition of this theory to it bare minimum with some additional background given in \cref{sec:notation_background}.
We call a subset of a measurable space $\measSp$ \definiendum{\SuslinText{}}, if it is the image of an element in $ \Borel{\mathbb N^{\mathbb N}} \measTimes \sAlg $ along the projection $ \mathbb N^{\mathbb N}\times\spSet \to \spSet$.
Let $\Souslin{\sAlg}$ denote the set of all \SuslinText{} subsets of $\measSp$.
A measurable space $\measSp$ is called \definiendum{analytic} if it is homeomorphic to a \SuslinText{} set of a standard space $\measSp*$ endowed with the \definiendum{restricted \salg}, i.e.\nolinebreak[3]\ $
    \sAlg*|_{\spSet} \coloneqq \setBuilder{B \cap \spSet}{ B \in \sAlg* }
$.
Such constructions are also a common subject in descriptive set theory, cf.\nolinebreak[3]\ \cite{Kechris12} and \cite[Ch.~42]{Fremlin}.
By $\Cat{Ana}$ we denote the full subcategory of $\Cat{Meas}$ of analytic spaces, which admits countable products.
The endofunctor $\giry$ restricts to 
\begin{CALCOversion}
    $\Cat{Ana}$.    
\end{CALCOversion}
\begin{arXivVersion}
    $\Cat{Ana}$ (as a direct consequence of \cref{prop:giry_sAlg_generators_measurable}).
\end{arXivVersion}
To provide full generality, cf.\nolinebreak[3]\ \cref{rem:generalisation_dirIm}, of our results, we also introduce the concept of a smooth space \cite{Falkner81}: $\measSp$ is \definiendum{smooth}, cf.~\cref{ssec:SouslinOp_smoothSp}, if for any other measurable space $\measSp*$ any projection to $\spSet*$ of a Borel (or equivalently \SuslinText{}) subset of $\measSp* \times \measSp$ is \SuslinText{}. Nevertheless, it is perfectly fine to assume analytic spaces throughout at least for the first reading.

Given a measure space (i.e., a measurable space with a measure) $\measdSp$ one may wish to extend $\sAlg$.
The $\meas$-completion of $\measSp$, $\overline{\sAlg}^{\meas} \supseteq \sAlg  $, is defined as the smallest $\sigma$-algebra
containing $\sAlg \cup \{ B \subseteq X \mid \exists A \in \sAlg. \; 
m(A) = 0 \mbox{ and } B \subseteq A \}$.
A way to describe $ \overline{\measSp}^{\meas} $, when $\meas$ is a probability measure, is as the set of all $A$ such that there are $A_-, A_+$ with $\meas(A_-) = \meas(A_+)$ and $A_- \subseteq A \subseteq A_+$.
The measure $\meas$ uniquely extends to a measure on $\overline{\sAlg}^{\meas}$.
Given a measurable space $\measSp$, the \definiendum{universal completion} of $\sAlg$ (denoted $\overline{\sAlg}$) is the intersection of all completions of $\sAlg$ with respect to any (probability) measure on $\measSp$.
The universal completion is quite big; especially it contains $\Souslin{\sAlg}$ (so every Suslin set is universally measurable).

\section{Markov decision processes}
    \label{sec:MDP}

In this section we are going to instantiate our categorical parameters (cf. Assumptions~\cref{ass:a1} - \cref{ass:a5}) in the setting of measurable spaces. We begin by recalling the definition of a Markov decision processes from \cite{FernsPanangadenPrecup11} and view them as coalgebras.

\begin{definition}
        \label{def:MarkovProcess}
    A (continuous) \emph{Markov reward process} (MRP) is a coalgebra $\gamma$ of type
    \[\measSp \to \giry\measSp \times \BorelSp{[0,1]} \in \Meas.\]
    In other words, $\gamma$ is given by a pair $(\gamma^0,\gamma^1)$ of maps satisfying the following properties:
    \begin{itemize}
        \item $\gamma^0(x)$ (for each $x\in X$) is a probability measure;
        \item $\gamma^0(\_)(U)\colon \spSet \to [0,1]$ (for each measurable set $U\in\sAlg$) is a measurable function; and 
        \item $\gamma^1$ is a measurable function.
    \end{itemize}
    Moreover, given a countable set $\actions$ of actions, we define a \emph{Markov decision process} (MDP) \cite{FernsPanangadenPrecup11} as a coalgebra $\gamma$ of type
    \[
    \measSp \to \prod_\actions \left(\giry \measSp \times \BorelSp{[0,1]}\right) \in \Meas.
    \]
    In other words, the map $\gamma(x)(a)$ (for each state $x\in X$ and action $a\in \actions$), is a Markov process. Henceforth we write $\gamma_{a,x}$ to denote $\gamma(x)(a)$; so, $\gamma_{a,x}^0$ corresponds to a probability measure and $\gamma_{a,x}^1$ corresponds to a `reward' at $x$ for an action $a$.
\end{definition}
Thus, the endofunctors of interest are the following:
\begin{itemize}
    \item $\BfunctorMP = \giry \times [0,1]$ whose coalgebras correspond to Markov reward processes
    \item $\BfunctorMDP = \prod_\actions \BfunctorMP$ whose coalgebras correspond to Markov decision processes.
\end{itemize}

\subsection{Fibrations induced universally/l.s.m.\ predicates}
    \label{subsec:pred}

Having fixed the type of systems, we now look into the issue of endowing a bifibration structure on the space of all Boolean/quantitative predicates. Consider the indexed category $\predicates \measSp = \Meas(\measSp,(2,\Borel{2}))$ of Boolean predicates, i.e. a predicate $p\in \predicates\measSp$ is a measurable function of type $X \to 2$. The reindexing functor $f^*$ (for a measurable function $f\colon \measSp \to \measSp*$) is given by the inverse image operation (since inverse image of measurable sets is measurable). It is well known (originally due to \SuslinText{} \cite{Suslin1917}) that Borel measurable sets even on standard spaces are not closed under direct images; thus as a result the left adjoint to reindexing functor cannot exist in general. Nevertheless if we weaken measurable sets to \SuslinText{} sets (which are equivalently analytic sets for analytic spaces, cf.\nolinebreak[3]\ \cite[421K]{Fremlin} and \cite[13.3iii)]{Kechris12}), then \SuslinText{} sets of analytic spaces are preserved by direct image onto analytic spaces, cf.\nolinebreak[3]\ \cite{BresslerSion64} for an in-depth discussion how to develop these concepts.

In lieu of the above discussion, we restrict our state spaces to be analytic, i.e. our working category for the remainder is the category $\ana$ of analytic spaces and measurable functions as morphisms.
This is, on the one hand, a bit more general than Polish spaces as required in \cite{FernsPanangadenPrecup11} and on the other hand conceptually more elegant, as we are only working with measurable spaces and do not require an underlying topology.
Moreover, we consider quantitative predicates on an analytic space $\measSp$ to be \emph{lower semimeasurable} (l.s.m.) functions (a real-valued generalisation of \SuslinText{} sets) of type $X \to [0,1]$.

\begin{definition}
    Let $\measSp$ be an analytic space. Then a function $p \colon X \to [0,1]$ is \definiendum{lower semi-measurable}, \definiendum{l.s.m.} for short, (resp.\nolinebreak[3]\ \definiendum{universally measurable}) predicate iff the preimage of the interval $[0,r]$ (for every $r\in[0,1]$) under $p$ is a \SuslinText{} set (resp.\nolinebreak[3]\ universally measurable set), i.e. for every $r\in [0,1]$, $p^{-1}([0,r]) \in \Souslin \sAlg$ (resp.\nolinebreak[3]\ $p^{-1}([0,r]) \in \overline \sAlg$).
\end{definition}

The term ``lower semi-measurable'' is chosen in parallel to the term ``lower semi-continuous'' in topology which refers to a real-valued function which is continuous with respect to the upper-interval topology $ \setBuilder{ (r, \infty) }{r \in \mathbb R} $.

We can arrange l.s.m.\nolinebreak[3]\ predicates in an indexed category as follows. Consider the mapping $\spredicates \colon \ana^\text{op} \to \pos$ such that $\spredicates \measSp$ is the set of all l.s.m.\nolinebreak[3]\ predicates where the ordering relation is the pointwise lifting of the `greater-than-equality' relation on the unit interval\footnote{In other words, we are viewing the unit interval as the Lawvere quantale $([0,1],\geq,+)$ where $+$ is the truncated addition. So, $p\leq q \iff \forall_{x}\ p(x) \geq q(x)$.}. The reindexing $f^*$ (for an arrow $f\colon \measSp \to \measSp* \in \ana$) is given by pre-composition, i.e.
\[
f^* (q) = q\circ f \qquad \text{(for every $q \in \spredicates \measSp*$).}
\]

\begin{lemma}
        \label{lem:directImage}
    The indexed category $\spredicates$ has countable fibred (co)limits, i.e. each fibre has countable meets and countable joins which are preserved by the reindexing functor. Moreover, $\spredicates$ has a bifibration structure, i.e., for every $f\colon \measSp \to \measSp*\in\ana$, the reindexing functor has a left adjoint $\exists_f$ given by:
    \[
    \exists_f (p) (y) = \inf_{f(x)=y} p(x) \qquad \text{(for every $p\in\spredicates\measSp,y\in Y$)}.
    \]
\end{lemma}
\begin{remark}
    \label{rem:generalisation_dirIm}
It should be noted that the existence of a left adjoint can be stated in more general terms by requiring that $\measSp$ is a smooth space, cf.
\cref{ssec:SouslinOp_smoothSp},
and $\measSp*$ is countably separated, i.e.\ there is a countable family of measurable sets distinguishing every pair of distinct points. 
Moreover, 
Axiom \ref{ass:a2} could be weakened to maps of type $f\colon B X\to Y\in \cat$:
  \begin{list}{
   \textcolor{lipicsGray}{\sffamily\bfseries\upshape\mathversion{bold}A\arabic{axiomctr}'.}
  }{%
    \usecounter{axiomctr}%
    \setlength{\leftmargin}{.33em} 
  }%
    \item
        the reindexing functor $f^*$ has a left adjoint
    $\exists_f\colon\spred {BX} \to \spred Y $.
\end{list}
Any universally measurable subset of a standard space is countably separated as a measurable space, but
\begin{arXivVersion}
    also---in virtue of \cref{cor:predLifting_semimeas}---an
\end{arXivVersion}
\begin{CALCOversion}
    also an
\end{CALCOversion}
analytic space.
So our construction generalises to the full subcategory of $\Cat{Meas}$ of measurable spaces expressible in this form.
\end{remark}

To show that \cref{ass:a2} is satisfied, it remains to define an indexed category $\predicates\colon \ana^{\text{op}} \to \pos$ such that each fibre $\spredicates\measSp$ is contained in $\predicates \measSp$. We disregard the trivial definition, i.e. $\predicates = \spredicates$, since \SuslinText{} sets are not closed under complementation. As a result, we cannot give semantics to the negation operator in our logic. Nonetheless, it is also known that \SuslinText{} sets are universally measurable sets \cite{Doberkat14}, so we simply let $\predicates\measSp$ be the set of universally measurable predicates on the analytic space $\measSp$.
\begin{proposition}
    Assumptions~\cref{ass:a1} and \cref{ass:a2} are satisfied by  $\predicates$ and $\spredicates$, respectively.
\end{proposition}


\section{Bisimulation distance}
    \label{sec:bdistance}
The objective of this section is to define bisimulation pseudometrics (cf. \cref{subsec:cliftingMDPs}) for Markov reward processes and MDPs as the least\footnote{Recall the predicates are ordered by $\geq$, so the greatest fixpoint is actually least fixpoint under the usual order $\leq$.} fixpoint of a functional given in \eqref{eq:gfp-clift} on page~\pageref{eq:gfp-clift} where $B=\{\BfunctorMP,\BfunctorMDP\}$. In both cases, the definition of a pseudometric rests on a coupling-based lifting \eqref{eq:clift} for the Giry endofunctor $\giry$ which we will work out in the following subsection.

\subsection{Wasserstein lifting categorically}\label{subsec:wlift}
We begin by defining a predicate lifting for $\giry$ (i.e.  when $B=\giry$ in Assumption~\cref{ass:a3}). Consider the mapping $\sigma_{\measSp} \colon \predicates \measSp \to \spred {\giry \measSp}$ given by
\begin{equation}\label{eq:ExpectationLifting}
\sigma_{\measSp} (p) (\meas) = \int p \diff \meas \quad (\text{for every $p\in\predicates{\measSp},\meas \in \giry\measSp$}).
\end{equation}
Henceforth, we drop the sigma-algebra notation from the subscript whenever it is clear from the context. Thus, $\sigma_X(p)(\meas)$ is the expectation of random variable $p$ under the measure $\meas$.

\begin{theorem}
        \label{thm:predLifting_semimeas}
    The mapping $\sigma$ defined in \eqref{eq:ExpectationLifting} is a natural transformation valued in $\spredicates$; thus, an indexed category morphism. Moreover, $\sigma$ preserves directed suprema which is a consequence of the monotone convergence theorem well known in measure theory.
\end{theorem}

Note that predicate lifting improves universally measurable predicates even to Borel measurable predicates for analytic spaces; the proof of this fact can be found in \cite{arXiv}.

Thus, \cref{ass:a3} is satisfied and invoking the $\hat\sigma$ given in \eqref{eq:clift} gives the usual Wasserstein lifting for the Giry endofunctor $\giry$ as expected. 
\begin{definition}
        \label{def:pmetric_coupling}
    A predicate $d\in \pred{\measSp \times \measSp}$ is a \emph{pseudometric} on $\measSp\in\ana$ iff $d$ is reflexive, symmetric, and satisfies the triangle inequality. 
    
    Moreover, a probability measure $\meas[c] \in \giry(\measSp \times \measSp)$ is a \emph{coupling} for two probability measures $\meas,\meas*$ iff $\giry(\proj 1)(\meas[c]) = \meas$ and $\giry(\proj 2)(\meas[c]) = \meas*$. We write $\couplings(\meas,\meas*)$ to denote the set of all couplings for the probability measures $\meas,\meas*$.
\end{definition}
\begin{proposition}\label{prop:clift-Giry}
Let $d$ be a pseudometric on a space $\measSp\in\ana$. Then, the lifting $\hat\sigma$ given in \eqref{eq:clift} evaluates to the following well known formula  associated with Wasserstein lifting of probability measures. Moreover, $\hat\sigma(d)$ is a pseudometric on $\giry\measSp$.
\[
\hat\sigma (d)(\meas,\meas*) = \inf_{\meas[c]\in K(\meas,\meas*)} \int d \diff \meas[c]
\]
\end{proposition}

\subsection{Distance lifting for \texorpdfstring{$B=\{\BfunctorMP,\BfunctorMDP\}$}{B}}\label{subsec:cliftingMDPs}

One way to define the distance lifting $\dliftMP$ for $\BfunctorMP$, i.e. a map of type
\[\dliftMP_X\colon \pred {\measSp \times \measSp} \Rightarrow \pred {\BfunctorMP \measSp \times \BfunctorMP \measSp}, \]
is to first define a predicate lifting for $\BfunctorMP$ and then use the equation in \eqref{eq:clift} where $B=\BfunctorMP$. To this end one may follow \cite[Subsection~5.6.2]{KerstanPhD} in deriving a predicate lifting for $\BfunctorMP$ in a compositional manner. These results (though stated for the category of sets) can be generalised to measurable spaces, but they are only applicable when the underlying endofunctor preserves weak pullbacks. In particular, it is known that the Giry functor (a composite functor in the case of $\BfunctorMP$) does not preserve weak pullbacks in $\Meas$ \cite{Viglizzo05}.

So instead of compositionally deriving predicate liftings for $\BfunctorMP$ and then invoking \eqref{eq:clift}, we derive the coupling based lifting for $\BfunctorMP$ in three stages:
\begin{itemize}
    \item first, we view $\BfunctorMP$ as the composition of functors $B_I \circ \giry$ where $B_I\colon \ana \to \ana$ maps every space to its product with the unit interval, i.e. 
    \[
    B_I = \text{Id} \times ([0,1],\Borel{[0,1]}).
    \]
    \item second, for a fixed $c \in [0,1]$, we define 
    \[ \sigma^c_X \colon \pred {\measSp \times \measSp} \Rightarrow \pred {B_I \measSp \times B_I \measSp}\] as
    $\sigma^c_X(d)((x,r),(y,s)) = cd(x,y) + (1-c)|r-s|$, for $x,y\in X$ and $r,s\in[0,1]$.
    \item third, recall $\hat\sigma$ from \cref{prop:clift-Giry} and let $\dliftMP$ be the composition: 
    \[
    \pred {\measSp \times \measSp} \xrightarrow{\hat\sigma_X} \pred {\giry \measSp \times \giry \measSp} \xrightarrow{\sigma_{\giry X}^c} \pred {\BfunctorMP \measSp \times \BfunctorMP \measSp }.
    \]
\end{itemize}
Note that $c$ may take the extremal values 0 and 1.
This is possible---in contrast to \cite{FernsPanangadenPrecup11}---as the bisimulation distance is not obtained using a contraction-based fixpoint argument.
However, in the extreme cases the bisimulation distance would not take into account either the transition or the reward part.

\begin{lemma}
        \label{lem:dliftingForMP}
    The above mapping $\sigma^c$ is well defined. Moreover, for any pseudometric $d\in \pred{\measSp \times \measSp}$, the lift $ \dliftMP_X(d) $ is given by
    \[
            \dliftMP_X(d)((\meas,r),(\meas*,s)) 
        =   c\left(\inf_{\meas[c]\in K(\meas,\meas*)} \int d \diff \meas[c]\right) + (1-c)|r-s|
    \]
    for $\meas,\meas*\in\giry \measSp$, and $r,s\in[0,1]$ and is a pseudometric.
\end{lemma}

In a similar vein, we can now define a distance lifting $\dliftMDP$ for $\BfunctorMDP$ (whose coalgebras model MDPs) by letting $\BfunctorMDP = B_\Sigma \circ \BfunctorMP$, where $B_\Sigma = \prod_\Sigma \text{Id}$. Now consider the distance lifting $\sigma^\Sigma_{(\blank)}$ for $B_\Sigma$ as follows:
\[
\sigma^\Sigma (d)(\vec x,\vec y) = \sup_{a\in\Sigma} d(\vec x(a),\vec y(a)), \quad \text{(for $\vec x,\vec y\in \prod_\Sigma \spSet$)}.
\]

\begin{lemma}\label{lem:dliftingForMDP}
    The above mapping $\sigma^\Sigma$ is well defined. Moreover, the mapping $\dliftMDP=\sigma^\Sigma_{\BfunctorMP (\blank) } \circ \dliftMP$ for a distance $d\in \pred{\measSp \times \measSp}$ evaluates to
    \[
    \dliftMDP(d)((\vec{\meas},\vec{r}),(\vec{\meas*},\vec{s})) =
    \sup_{a\in\Sigma} \left[ c\left(\inf_{\meas[c]\in K(\vec\meas(a),\vec{\meas*}(a))} \int d \diff \meas[c]\right) + (1-c)|\vec r(a)-\vec s(a)|\right].
    \]
\end{lemma}
Now composing the two distance liftings $\sigma^B$ (for $B\in\{\BfunctorMP,\BfunctorMDP\}$) with a $B$-coalgebra is the desired functional as given in \eqref{eq:gfp-clift}. Clearly, \cref{ass:a5} is satisfied since $\predicates$ has countable suprema and they are preserved by reindexing functors. We end this subsection by showing that the least fixpoint exists for both of these functionals; thus, also paving a way to compute bisimulation pseudometrics for these systems. To this end we need a general result, whose proof is based on some classical results comprising a non-topological version of Riesz–Markov–Kakutani representation theorem \cite[IV.5.1]{DunfordSchwartz58}, Banach-Alaoglu theorem \cite[V.4.2]{DunfordSchwartz58} and Sion's minimax theorem \cite[Thm.~3]{Simons95}.

\begin{theorem}
        \label{thm:WassersteinDist_cpoCont}
    Let $\gamma\colon \measSp \to \giry \measSp \in \ana$. 
    Then the functional $\gamma \circ \hat\sigma$ is $\omega$-cpo continuous w.r.t. $\leq$, i.e. for any $\leq$-increasing sequence $d_i \in \pred {\measSp \times \measSp}$ of pseudometrics with $i\in\nat$, we have (for each $x,y\in X$):
\[
  \inf_{\meas[c] \in K(\gamma_x,\gamma_y)}
            \int \sup_{i \in \nat} d_i \diff \meas[c]  
        =
            \adjustlimits
            \sup_{i \in \nat} 
            \inf_{\meas[c] \in K(\gamma_x,\gamma_y)} 
            \int d_i \diff \meas[c]
    \text.    
\]
\end{theorem}

\begin{remark}
    The above theorem can be stated for general measurable spaces as well, but by restricting the coalgebra map so that $\gamma(x)$ (for each $x\in X$) is a perfect measure (see \cref{ssec:perfectMeasures}). Perfect measures were introduced by Kolmogorov \cite[22--23]{GnedenkoKolmogorov49} and have many different equivalent definitions. For us
    a measure space $\measdSp$ is called \definiendum{perfect}, if for any separable metrisable space $(\spSet*, \topy) $ and every measurable map $f\colon \spSet\to\spSet*$ we have the following property: For every $A \in \sAlg $ and $ r < \meas(A) $ there is a compact set $K \subseteq \image f$ with $\meas( A \cap \invImSet f K ) \geq r$, cf.\nolinebreak[3]\ \cite[451O(a)]{Fremlin}. In this case, the measure $\meas$ is called \definiendum{perfect}. 
\end{remark}

\begin{remark}
        \label{cor:WassersteinDist_p_cpoCont}
    Let $p\in [1, \infty)$ and recall\hbtodo{Check this} the $p$-Wasserstein distance between probability measures $\meas,\meas*\in\giry \measSp$. Below we argue how to capture this lifting in our setup. 
    \[
        W_p(\Dist*)(\meas,\meas*)
        =      
        \sqrt[\leftroot{-1}\uproot{3}{\textstyle p} ]{\inf_{\meas[c] \in K(\meas,\meas*)} 
                    \int {\Dist*}^p \diff \meas[c]
                }, 
        \qquad \text{for a pseudometric $d\in\pred{\measSp \times \measSp}$}.
    \]
    Note that any monotonously increasing lower semicontinuous function $f\colon [0, 1] \to [0,1]$ induces a \ocpo-continuous map $ 
        f\circ \blank \colon \predicates\measSp \to \predicates\measSp 
    $. It is monotone by monotonicity of $f$ and for any increasing sequence $(p_i \in  \predicates \measSp)_{i\in\nat}$ we have
    \[
    \sup_{i\in\nat} f (p_i (x)) = f \left(\sup_{i\in\nat} p_i(x) \right) \qquad \text{(for each $x\in X$)}.
    \]
    Note that both $(\blank)^p $ and $ \sqrt[\leftroot{1}\uproot{3}p]{\vphantom l \blank} $ are monotonously increasing lower-semicontinuous functions.
    So $W_p = \sqrt[\leftroot{1}\uproot{3}p]{\vphantom l \blank} \comp \hat \sigma \comp (\blank)^p $ is \ocpo-continuous as a composition of \ocpo-continuous functions.
\end{remark} 

\begin{corollary}\label{cor:bdistance-exists}
Let $\gamma_B\colon \measSp \to B\measSp \in \ana$ be a coalgebra where $B\in\{\BfunctorMP,\BfunctorMDP\}$. Then the least fixpoint for the functionals $\gamma_B \circ \sigma^B$ exists.
\end{corollary}

Using the fact that any \ocpo-continuous endofunction has a least fixpoint by Kleene's fixpoint theorem, we write $\bdistance_c^\gamma$ (or simply $\bdistance$ whenever the coalgebra structure is clear from the context) to denote the least fixpoint of the functionals in the above corollary. 


Note that by using Kleene fixed point theorem we require weaker assumptions than \cite[3.12]{FernsPanangadenPrecup04}, who use the Banach fixed point theorem to define their bisimulation pseudometric restrict themselves to a set-up with contractions.



\section{A quantitative modal logic and its expressivity}
    \label{sec:log-MDPs}
The signatures of the (logical) languages considered in this paper are parametrised by a set $ \setBuilder{f_{iy}}{ y \in Y_i } $ of $\omega$-indexed families of $Y_i$-indexed function symbols of arity $n_i$ as follows:
\begin{equation}
            \label{eq:Sprache_def}
        \top 
        \mid \neg \blank
        \mid \blank \land \blank 
        \mid\diamond_a \blank, a \in A
        \mid f_{iy}, y \in Y_i, i \in \omega
\end{equation}
In other words, the signature we are using extends the semi-lattice signature ($\top$, $\land$) with negation ($\neg$), modalities $\diamond_a$ (one for each action $a$) and additional function symbols $f_{iy}$ (each $f_{iy}$ could be viewed as $n_i$-ary predicate on the unit interval).
We simply write $\SpracheFormulas $ to denote the set of formulae generated by the above signature. 
The restriction to countably many families of function symbols will become important when we construct a second-countable topology on $\SpracheFormulas$.
Note that we can also view the logical symbols, $\top, \neg\blank, \blank\wedge\blank $, as (singleton index families of) function symbols. This will be very handy for proofs by structural induction over $\Sprache$. 
Throughout this section, we let $\Omega=[0,1]$ and consider only MDPs (the modal logic for MRPs can be derive by letting the set $\actions$ of actions to be a singleton set). 

The reason to choose this general formulation with index sets $Y_i$ are twofold. First, to endow a topology on $\Sprache$ which is needed to prove both adequacy (i.e. the distance induced by formulae in $\Sprache$ is below than the bisimulation distance $\bdistance$ for MDPs) and expressivity (which is the converse of adequacy). Second, it allows greater flexibility 
in applications, cf.\nolinebreak[3]\ \cref{eq:basicSymb_interpret}, by accommodating additional operators for which adequacy and expressivity still hold. 


\subsection{Interpretation of modal formulae in \texorpdfstring{$\Sprache$}L}
    \label{ssec:interpretation}
We give semantics to each formulae $\varphi \in \Sprache$ by defining an \definiendum{interpretation $\interpret{\varphi}$} as a predicate in $\predicates \measSp$. This is done by structural induction over $\varphi$ in the following way.
The logical symbols, i.e.\nolinebreak[3]\ 
$\top$ (truth), $\neg \blank$ (negation), and $\blank \land \blank$ (conjunction),
respectively, are interpreted as the functions
\begin{equation}
        \label{eq:logicSymb_interpret}
    1, \quad
    \lambdCalc x. 1 -x, \quad
    \lambdCalc x\, y.\min\{x, y\}\quad
\text{respectively.}
\end{equation}
Next we consider function symbols $f_{iy}$ of arity $n_i>0$ (for some $y\in Y_i$) and a sequence $(\varphi_j)_{1\leq j \leq n_i}$. Its interpretation is given as 
\[\interpret{f_{iy}}(\varphi_1,\cdots,\varphi_{n_i}) = f_{iy}(\interpret{\varphi_1},\cdots,\interpret{\varphi_{n_i}}) 
\mbox{ for some fixed } f_{iy}: [0,1]^{n_i} \to [0,1] \]
whilst, $\interpret{f_{iy}} = f_{iy}\in [0,1]$ when $n_i=0$ (i.e. $f_{iy}$ is a constant). 
Besides the logical symbols we introduce the following \definiendum{basic} families of function symbols indexed by $r,c\in [0,1]$ and $\alpha\in (0,\infty)$:
\definiendum{Scalar addition} ($\blank + r$), \definiendum{scalar subtraction} ($\blank-r$), \definiendum{scalar multiplication} ($r\blank$), and \definiendum{convex combination} ($\blank +_c \blank$) are interpreted as the function assigning to $x$ (and $y$)
\begin{equation}
    \label{eq:basicSymb_interpret}
    \min\{1, x+r\}, \quad
    \max\{0, x-r\}, \quad
    rx, \quad
    x +_c y
\text,
\end{equation}
respectively, throughout this paper.
Note that only the first two basic operations are required for our main results.

Finally, for the modal operators we fix a parameter $c\in\Omega$ (which was also used to define bisimulation pseudometrics $\bdistance$ in the previous section) and a coalgebra $\gamma\colon \measSp \to \BfunctorMDP \measSp \in \ana$ to define the interpretation of $\interpret{\diamond_a \varphi}$ as:
\begin{equation}
        \label{eq:modalOp_explicitly}
        \interpret{\diamond_a \varphi}_\gamma(x)
    =  
    \int \interpret{\varphi} \diff \gamma_{a,x}^0 \ +_c\ \gamma_{a,x}^1,
\end{equation}
where $ r +_c s \coloneqq cr + (1-c)s $ denotes convex combination of $r,s\in\Omega$. Intuitively, this means that the value of $\diamond_a \varphi$ is determined by a convex combination of the expected value of $\varphi$ and the utility after performing $a$. 

The defined interpretation function $\interpret{\blank}$ gives rise a (quantitative) theory map $\quantTheory\colon\spSet \to \Omega^\Sprache$ defined by
\begin{equation}
        \label{eq:quantitativeTheory_def}
        \quantTheory(x)(\varphi)
    \coloneqq
         \interpret{\varphi}(x) \quad \text{(for every $\varphi\in\Sprache$)}.
\end{equation}
The question whether the theory map lives in our working category $\ana$ is an important step for expressivity (cf.\nolinebreak[3]\ \cref{thm:expressivity}). However, for adequacy (cf.\nolinebreak[3]\ \cref{thm:adequacy}), we only require that the function symbols $f_{iy}$ in $\Sprache$ are nonexpansive w.r.t.\nolinebreak[3]\ the suprema distance $d_\infty(\vec x,\vec y) = \max_{1\leq j \leq n_i} \abs{\vec x(j),\vec y(j)}$. Nonetheless, before attempting these results we need the definition of a logical distance, which at this stage is purely a set-theoretic assignment. Later, in next subsection, we will show that the logical distance defined below is indeed a predicate over the product of a state space with itself (cf.\nolinebreak[3]\ \cref{lem:logicalDist_measurable}). 

\begin{definition}[Logical distance]
Given an interpretation $\interpret\blank\colon \SpracheFormulas \to \set( X,[0,1]) $, we define the \definiendum{logical distance} as follows (for every $x,y\in X$): 
\begin{equation}
            \label{eq:dist_Sprache_def}
        \ldistance(x,y) 
    \ =\   \sup_{\varphi \in \Sprache }\ \left| \interpret{\varphi}(x) - \interpret{\varphi}(y) \right| 
    \ =\  \sup_{\varphi \in \Sprache }\  \interpret{\varphi}(x) - \interpret{\varphi}(y) \enspace. 
\end{equation}
Note that the absolute value is redundant since the negation operator $\neg$ is in our logic $\Sprache$. 
\end{definition}

\subsection{Endowing a topology on \texorpdfstring{$\Sprache$}L through its shapes}\label{ssec:shapes}
As we are dealing with more than countably many function symbols, we will have to impose some structure on the set of function symbols in order to prove our expressivity theorem (cf.\nolinebreak[3]\ \cref{thm:expressivity}). 
Technically, we need a topology $\topy_\Sprache$ on $\Sprache$ so that the theory map $\quantTheory$ \eqref{eq:quantitativeTheory_def} becomes \emph{topologisable} in the following sense.

\begin{definition}
Let $\measSp \in \ana$. Then the theory map $\quantTheory\colon X \to \Omega^\Sprache$ is \emph{topologisable} iff there is a topology $\topy_\Sprache$ on $\Sprache$ such that the preimage of every open set $U\subseteq \Omega^\Sprache$ (in the compact open topology on the function space $\Omega^\Sprache$) is a universal measurable set, i.e. $\invImSet{\quantTheory} (U) \in \overline{\sAlg}$.
\end{definition}
To be able to do this, we switch our focus from the uncountable language  $\SpracheFormulas$ to the countable language $\Shape(\Sprache)$ of \emph{shapes} induced by the language $\Sprache$. This language $\Shape(\Sprache)$ of shapes is constructed 
by collapsing all function symbols indexed by the same index set, i.e.\
\begin{equation}
            \label{eq:SpracheShape}
        \top 
        \mid \neg \blank
        \mid \blank \land \blank 
        \mid\diamond_a \blank, a \in A
        \mid f_{i}, i \in \omega
\text.
\end{equation}
To each formula $\varphi$ on $\SpracheFormulas$ one can assign a formula in $\Shape(\Sprache)$ by replacing each $f_{iy}$ by $f_i$.
This defines an equivalence relation on $\Sprache$.
For each $\psi \in \Shape(\Sprache)$
denote by $\widehat{\psi} \subseteq \SpracheFormulas $ the corresponding equivalence class.
For instance, for an $n_i$-ary function symbol $f_i$ the set $\widehat{\psi}$ for $\psi = f_i(\top, \ldots, \top)$ is in canonical bijection to $Y_i$,
and the shape $ r\neg\top \wedge r\top $ (when $\Sprache$ allows for scalar multiplication) corresponds\cktodo{it's not {\em clear} in what sense $r \top$ is a shape. what is the formal status of $r$ in that expression?} to $\Omega \times \Omega$, as each $r$ ranges over $\Omega=[0,1]$.
Through this equivalence relation we can subdivide $\SpracheFormulas$ into countably many chunks, each of which associated to a finite product of $Y_i$'s. 
If for a family $ \{f_y\}_{y \in Y} $ of $n$-ary function symbols the set $Y$ is endowed with a \salg\  (resp.\nolinebreak[3]\ topology), we call $\interpret{\blank}$ \definiendum{jointly measurable} (resp.\nolinebreak[3]\ \definiendum{jointly continuous}), if the function
$
    [0,1]^n \times Y \to [0,1]
$
is measurable (resp.\nolinebreak[3]\ continuous) with respect to the respective product \salg{} (resp.\nolinebreak[3]\ product topology). 


For the remainder of this subsection $\Omega^{\topSp }$ will denote the set of continuous functions from $ \topSp  $ to $\Omega$, which will be endowed with the compact-open topology \cite[3.4]{Engelking89} if not explicitly stated otherwise.
Recall that a Hausdorff topological space $\topSp $ is called \definiendum{locally compact}, if each point admits a compact neighbourhood.

\begin{lemma}{\cite[3.4.16]{Engelking89}}
        \label{lem:FctSp_2ndCount}
    Let $(X, \topy )$ be a locally compact second countable space.
    Moreover if $\Omega$ is second countable,  then $ \Omega^{(X, \topy )} $ is second countable w.r.t.\nolinebreak[3]\ the compact-open topology.
\end{lemma}


\begin{remark}
        \label{rem:lem:FctSp_2ndCount}
    \Cref{lem:FctSp_2ndCount} will be used only once but at a crucial point in \cref{thm:expressivity}, which is our main result.
    One could ask if one can extend the classes of spaces for which \cref{lem:FctSp_2ndCount} holds.
    For instance, does it hold for arbitrary second countable Hausdorff spaces?
    An ansatz would be to weaken the notion of compactness further.
    The promising notion is a k-space, the quotient space of some locally compact space, i.e.\nolinebreak[3]\ there is a quotient map $ q\colon \topSp* \to \topSp $ for some locally compact space $\topSp$ \cite[p.~152]{Engelking89}.
    The reason is that in this case there is a decomposition $ \Omega^{\topSp*} = \lim\limits_{\longleftarrow} {}_{i} \Omega^{K_i} $ for any directed system of of compact sets $K_i \subseteq \spSet$ with $\bigcup_i K_i = \spSet$ \cite[3.4.11]{Engelking89}.
    Second countability follows as soon as $ i $ can be chosen to range over a countable set.
    Unfortunately, this already implies hemicompactness of $\topSp*$ in case of regular spaces, cf.\nolinebreak[3]\ \cite[8.1(d)]{Douwen84} and \cite[3.4.E(c)]{Engelking89}.
    \cite{Douwen84} also discusses other weakenings of compactness, but the mentioned Fact~8.1 therein 
    do not provide a remedy.
\end{remark}


We end this subsection by the main results of this section; namely that the logical distance $\ldistance$ is a predicate on $\measSp \times \measSp$ and the logic $\Sprache$ is adequate w.r.t.\nolinebreak[3]\ $\bdistance_c$.
\begin{lemma}
        \label{lem:logicalDist_measurable}
    If $\quantTheory$ on $\measSp$ is topologisable, 
    then $\ldistance$ is measurable. 
\end{lemma}



\begin{theorem}
        \label{thm:adequacy}
    If the function symbols $f_{iy}$ are interpreted by $ \interpret{\blank} $ as nonexpansive functions w.r.t.\nolinebreak[3]\ $\Dist[\infty]*$, then the logic $\Sprache$ is adequate, 
    i.e. $\bdistance \geq \ldistance$. 
    
    Moreover, the language consisting of logical symbols, modalities, scalar addition, subtraction, multiplication, and convex combination is always adequate.
\end{theorem}

\begin{remark}
With the aim to develop continuous version of first order logic, Yaacov and Berenstein studied metric structures and their model theory in \cite{ModelTheoryMet:2008}. Simply put, a metric structure consists of a complete bounded metric space $(M,d)$ with a set of $\mathbb R$-valued predicates, a set of operators on the metric space (which are uniformly continuous of type $M^n \to M$ for some $n>0$), and a set of distinguished elements of $M$. It is worthwhile to note that the operators defined by our signature \eqref{eq:Sprache_def} is a special case of a metric structure on the unit interval with an empty set of predicates. 
\end{remark}

\subsection{A general expressivity theorem for \texorpdfstring{$\Sprache$}L}
    \label{ssec:expressivity}

Our main expressivity result (cf.\nolinebreak[3]\ \cref{thm:expressivity}) is based on the well known Stone-Weierstra\ss\ theorem and Kantorovich-Rubinstein duality. This follows the tradition of expressivity results from recent papers \cite{KomoridaKatsumataKu-ExpressivityofQuant,beohar_et_al:LIPIcs.STACS.2024.10,forster_et_al:LIPIcs.CSL.2023.22} on coalgebraic modal logic; however, the proof of \cref{thm:expressivity} does not follow from the abstract results established in the aforementioned articles. In contrast, we need an extra condition that the theory map $\quantTheory$ is topologisable (cf.\nolinebreak[3]\ \cref{lem:logicalDist_measurable}).

\begin{definition}
        \label{def:approx_pair_pt}
    A set $L$ of functions $\spSet \to [0,1]$ \definiendum{approximates a function $f\colon \spSet \to \mathbb R$ at a pair $x,y \in \spSet$ up to $\varepsilon$} (for some $\varepsilon > 0$) if there is a $g \in L$ with $ \abs{g(x)-f(x)}, \abs{g(y)-f(y)} < \varepsilon $.
    We further say that $L$ \definiendum{approximates $f$ at }$x, y$ if $L$ approximates $f$ at $x, y$ for each $\varepsilon > 0$.
\end{definition}

Henceforth we write $\interpret{\Sprache} = \setBuilder{\interpret{\varphi}}{\varphi \in \Sprache}$. It turns out that the operators (truth, conjunction, positive, and negative scaling) in our logic $\Sprache$ approximate any predicate over a state space. The following lemma is taken from \cite[Lemma~10]{ChenClercPanaganden25}. 

The following lemma is a continuous (and thus simpler) version of \cite[Lem.~10]{ChenClercPanaganden25}.
\begin{lemma}
        \label{lem:approxAtPairOfPoints}
Assume $
        \top, \blank\wedge\blank, \blank+r, \blank-r\in \Sprache
    $ for all $ r \in \Omega $ and $\interpret\blank\colon \SpracheFormulas\to \predicates\measSp$ an interpretation on predicates on some measurable space $\measSp$.
    Let $ p \in \predicates\measSp $, $ x,y \in \spSet $ and $\varepsilon >0$.
    \begin{subequations}
    Then
    \begin{align}
        \label{eq:lem:approxAtPairOfPoints_assms}
    \MoveEqLeft
            \exists \varphi \in \Sprache\colon
                0 \leq p(x) - p(y) < \interpret{\varphi}(x) - \interpret{\varphi}(y) + \varepsilon
    \\  \label{eq:lem:approxAtPairOfPoints_approxExact}
        &\implies
            \bigl(
                \exists \psi \in \Sprache\colon
                p(x) - \interpret{\psi}(x) \in [0,\varepsilon ) \text{ and } p(y) - \interpret{\psi}(y) = 0
            \bigr)
    \\  \label{eq:lem:approxAtPairOfPoints_approx}
        &\implies
            \interpret\Sprache \text{ approximates $p$ at $x,y$ up to }\varepsilon
    \text.
    \end{align}
    \end{subequations}
\end{lemma}

With the help of this lemma, we can approximate any short (aka nonexpansive) predicates w.r.t. logical distance $\ldistance$ by formulae in our logic $\Sprache$. Given a measurable space $\measSp$, we define the set of \definiendum{short predicates with respect to logical distance $\ldistance$} be
\begin{equation}
        \label{eq:predicates_short}
        \predicates(\spSet, \sAlg, \ldistance)
    \coloneqq
        \setBuilder
            {h \in \predicates\measSp}
            {\forall x, y \in\spSet\colon
                \ldistance (x,y) \geq h(x) - h(y)}
\text.
\end{equation}

\begin{corollary}
        \label{cor:approxAtPairOfPoints}
    As long as the scalar addition is allowed in the signature of $\Sprache$ in \eqref{eq:Sprache_def}, every short predicate can be approximated by the interpretation of modal formulae in $\interpret{\Sprache}$.
\end{corollary} 
\begin{proof}
    Take $h \in \predicates(\spSet, \sAlg, \Dist[\interpret{\blank}]*) $.
    As $ \neg\blank\in \Sprache $ we have for each $x,y \in \spSet$ by \cref{eq:dist_Sprache_def} that $
        h(x) - h(y) \leq \sup_{f \in \interpret\Sprache} f(x) - f(y)
    $.
    This implies \cref{eq:lem:approxAtPairOfPoints_assms} for any $\varepsilon > 0$.
    Thus $\interpret\Sprache$ approximates $h$ at $x, y$ for any $\varepsilon>0$ by \cref{lem:approxAtPairOfPoints}.
    Hence the claim follows.
\end{proof}

The next lemma is the well-known Kantorovich-Rubinstein duality extended to perfect measures. 
We provide a proof in the appendix, since the rather sketchy proof in \cite[Thm.~5]{RamachandranRüschendorf95} considers only distances on their induced Borel-\salg, while other known proofs, especially \cite[11.8.2\&{}6]{Dudley02}, chose a topological instead of a purely measure theoretic set-up.

\begin{lemma}[Kantorovich-Rubinstein theorem]
        \label{lem:KantorovicRubinstein}
    Let $\meas, \meas*\in \giry\measSp$ be perfect measures and $d\colon \measSp \times \measSp \to ([0,1],\Borel{[0,1]})$ a l.s.m.\nolinebreak[3]\ pseudo-metric such that $\measSp$ is analytic (or smooth, or $\topy_d$, the topology induced by $d$, is contained in $\overline{\sAlg}$) and $\Dist*(\blank,x_0)$ is integrable for every $ x_0 \in \spSet $ .
    Then
    \[
            \inf_{\meas[c]\in K(\meas,\meas*)} \int \Dist* \diff \meas[c]
        =   \sup_{h \in \predicates(\spSet, \sAlg, \Dist*)}
            \int h \diff (\meas - \meas*)
    \text.
    \]
\end{lemma}

\begin{theorem}
        \label{thm:expressivity}
    Let $\Sprache$ be a language with a coalgebra $\gamma \colon \measSp \to \BfunctorMDP \measSp$ so that $\interpret{\blank}$ is well defined. Assume the following restrictions:
    \begin{enumerate}
        \item \label{lst:thm:expressivity_nonTrivIneq_perfectRegular} every measure $\gamma_{x,a}^0$ (for every $x\in X$, $a\in\actions$) is perfect,
        \item the theory map $\quantTheory$ is topologisable, and
        \item the scalar addition is in the signature of our language $\Sprache$,
    \end{enumerate}
    then $\ldistance$ is a fixpoint of the functional $\sigma^{\BfunctorMDP} \circ \gamma$. As a consequence, we have that the language $\ldistance$ is expressive w.r.t.\nolinebreak[3]\ $\bdistance$ (i.e., $\bdistance\leq \ldistance$).
\end{theorem} 
\begin{proof}
    Let $\meas_{x,a} = \gamma_{x,a}^0$ (for each $x\in X,a\in \actions$), let $r^{a}_{x}=\gamma_{a,x}^1$ and  $r^a_{x,y} = \abs{r^a_x - r^a_y}$. Recall the distance liftings $\hat\sigma$ and $\dliftMP$ from \cref{prop:clift-Giry} and \cref{lem:dliftingForMP}, respectively.
    The claim $ \bdistance\leq \ldistance $ is---using order preservation of $\sigma^{\BfunctorMDP} \circ \gamma$, \cref{thm:WassersteinDist_cpoCont}---equivalent to
    $
        \forall \varepsilon > 0\colon
        \sigma^{\BfunctorMDP} \circ \gamma (x,y) \leq \ldistance(x,y)
        + \varepsilon
    $
    for all $x,y \in \spSet$.
    From the definition of $ \sigma^{\BfunctorMDP} $ this translates to the condition
    \begin{equation}
            \label{eq:thm:expressivity_proof}
    \forall \varepsilon > 0\colon
    \forall a \in \actions \colon\qquad
               \dliftMP_X(\ldistance) \circ \gamma (x,y) ((\meas_{x,a},r^a_x),(\meas_{y,a},r^a_y))
        \leq
            \ldistance + \varepsilon.
    \end{equation}
    To this end, we begin by expanding the left hand side of the above inequality:
    \begin{align}
    \nonumber
    \MoveEqLeft[4]
              \dliftMP_X(\ldistance) \circ \gamma (x,y) ((\meas_{x,a},r^a_x),(\meas_{y,a},r^a_y))
    \\
        &=
                \inf_{\meas[c]\in K(\meas_{x,a},\meas_{y,a})}
                    \int\ldistance \diff \meas[c]
            +_c r^a_{xy}
    \intertext{we push this expression to $[0,1]^{\SpracheFormulas}$ by letting $ \widetilde{\ldistance} \coloneqq \sup_{\varphi \in \image ( \quantTheory )}$ be a distance on $[0,1]^{\SpracheFormulas}$}
        &=
                \inf_{\meas[c]\in K(\quantTheory_*(\meas_{x,a}),\quantTheory(\meas_{y,a}))}
                    \int \widetilde{\ldistance} \diff \meas[c]
            +_c r^a_{xy}
    \intertext{As $[0,1]^\SpracheFormulas$ is second countable, cf.\nolinebreak[3]\ \cref{lem:FctSp_2ndCount}, we can (depending on $ \quantTheory_*(\meas_{x,a})$ and $\quantTheory(\meas_{y,a}) $) restrict the integral to $A \times A$ for some standard (and thus analytic) subspace $A \in \Borel{[0,1]^\SpracheFormulas} $ \cite[Thm.~6]{Faden85} (using that $ \quantTheory_*(\meas_{x,a})$ and $\quantTheory(\meas_{y,a}) $ again are perfect).
    (Note that this argument actually requires only second countability of $\SpracheFormulas$ as it can be refined using \cite[3.8.D]{Engelking89}, \cite[2.1.15]{Bogachev98} and \cite[\S~8~Remark]{Faden85}.)
    As $ \ldistance  $ is bounded,
        so $\ldistance(\blank, x)$ is certainly integrable for any $x\in \spSet$.
    Thus we can finally apply Kantorovic-Rubinstein duality \cref{lem:KantorovicRubinstein}, using \cref{lem:logicalDist_measurable} and Item \ref{lst:thm:expressivity_nonTrivIneq_perfectRegular} }
        &=\sup_{h \in \predicates([0,1]^{\SpracheFormulas}, \Borel{[0,1]^{\SpracheFormulas}}, \widetilde{\ldistance})}
                \int h \diff (\meas_{x, a} - \meas_{y, a})
                +_c r^a_{xy}
    \\
    \label{eq:thm:expressivity_proof2}
        &\leq
                    \sup_{h \in \predicates(\spSet, \sAlg, \ldistance)}
                \int h \diff (\quantTheory_*(\meas_{x, a}) - \quantTheory_*(\meas_{y, a}))
                +_c r^a_{xy}
    \end{align}
    By applying \cref{cor:approxAtPairOfPoints} we can approximate any short predicate $h$ by the interpretation of formulae in our logic $\interpret{\Sprache}$.

    Define $ \paving S \subseteq \sAlg $ to be
    $
            \paving S
        =
            \setBuilder*
                {U, \complOp{U}*}
                {U \in \invImSet{\quantTheory}*[\paving S_{\interpret\blank}] \cup
                \invImSet h * [\paving S_\Omega ] }
    \text.
    $
    Let $ \topy  $ denote the topology generated by $\paving S$ and $ \topSp[']* $ the Kolmogorov quotient, cf.\nolinebreak[3]\ \cref{sec:topology}, of $\topSp $ with unit map $\eta \colon \topSp \to \topSp[']*$.
    Both $
        \meas_{z,a}^\flat \coloneqq \giry(\eta)(\meas_{z,a})$ (for $z\in\{x,y\}$) are perfect since the push-forward measures of perfect measures \cite[451Ea]{Fremlin} is perfect and using \cref{lst:thm:expressivity_nonTrivIneq_perfectRegular} we know $\meas_{z,a}$ is perfect.
    By \cite[451M]{Fremlin} both measures are inner regular with respect to compact sets;
    thus we find compact sets $K_x, K_y \subseteq \spSet' $ with $
             \meas_{x,a}^\flat(\complOp{K_x}),    \meas_{y,a}^\flat(\complOp{K_y})
        <   \delta
    $.
    Thus $K_x \cup K_y $ is compact and so is $ K \coloneqq \invImSet{\eta}(K_x \cup K_y) $ by \cref{eq:Kolmogorov_unit_proper}.
    Moreover we have $
        \meas_{x,a}(\complOp{K}), \meas_{y,a}(\complOp{K_y}) < \delta
    $.
    As $\paving S$ is closed under complement, $ \topSp[']* $ is $\mathrm R_2$.

    So finally, the Stone-Weierstraß Theorem, \cref{lem:StoneWeierstraß}, is applicable to $ (K, \topy|_K) $; thus, every nonexpansive predicate $h$ can be approximated on $K$ by a function from the family $\setBuilder{\interpret{\varphi}|_K}{\varphi \in \SpracheFormulas}$. Let $\varphi_{h,\delta} \in \SpracheFormulas$ denote a witness of a $\delta$-approximation from $\interpret{\Sprache}|_K$ of $h|_K$.

    \allowdisplaybreaks
    Continuing at \cref{eq:thm:expressivity_proof2} we obtain (for all $\delta > 0$):
    \begin{align}
    \nonumber
    \MoveEqLeft[2]
                \sigma (\ldistance)(\meas_{x,a},\meas_{y,a})
            +_c r^a_{x,y}
    \\ \nonumber
        &\leq
            \biggl(\sup_{h \in \predicates(\spSet, \sAlg, \ldistance )}
                    \int_K h \diff (\meas_{x, a} - \meas_{y, a})
                +   \int_{\complOp K*} h \diff (\meas_{x, a} - \meas_{y, a})
            \biggr)
            +_c r^a_{xy}
    \\ \nonumber
        &\leq
            \biggl(\sup_{h \in \predicates(\spSet, \sAlg, \ldistance )}
                    \int_K \interpret{\varphi_{h,\delta}}
                        \diff (\meas_{x, a} - \meas_{y, a})
                +   \delta
                +   \delta \cdot 1
            \biggr)
            +_c r^a_{xy}
    \\ \nonumber
        &\leq
            \biggl(\sup_{h \in \predicates(\spSet, \sAlg, \ldistance )}
                            \int \interpret{\varphi_{h,\delta}} \diff (\meas_{x, a} - \meas_{y, a})
                +   3\delta
            \biggr)
            +_c r^a_{xy}
    \\ \nonumber
        &\leq
            \biggl(
            \sup_{\varphi \in \SpracheFormulas}
                \int \interpret\varphi \diff (\meas_{x, a} - \meas_{y, a}) + 3\delta
            \biggr)
            +_c r^a_{xy}
    \\ \label{eq:signature_to_be_applied}
        &=
            \biggl(
            \sup_{\varphi \in \SpracheFormulas}
                \left(\int \interpret\varphi \diff \meas_{x, a} +_c r^a_x \right)
                -
                \left(\int \interpret\varphi \diff \meas_{y, a}  +_c r^a_y\right)
            \biggr)
             + 3c\delta
    \\ \nonumber
        &=
            \sup_{\varphi \in \SpracheFormulas}
                    \interpret{ \diamond_a \varphi }(x)
                -   \interpret{ \diamond_a \varphi }(y)
                +   3c\delta
    \\ \nonumber
        &\textrel{\cref{eq:dist_Sprache_def}}\leq
            \ldistance (x,y) + 3c\delta
    \text.
    \end{align}
    Choosing $3c\delta < \varepsilon$, \cref{eq:thm:expressivity_proof} follows---finishing the proof.
    \allowdisplaybreaks[1]
\end{proof}

It should be noted that the restrictions on perfect measures in the above theorem is redundant when the coalgebra map $\gamma \in \ana$. Moreover, the second restriction from the previous theorem can also be discarded by imposing the following restrictions on the function symbols $(f_{iy})_{i\in\omega,y\in Y_i}$, which belong to the signature of our language $\Sprache$.



\begin{theorem}
        \label{thm:interpretation_hemicompact}
    Assume that $\Sprache$ given in \eqref{eq:Sprache_def} is such that
    each family $Y_i$ is endowed with a second countable Hausdorff topology $ \topy_i $ and the interpretation of function symbols $f_{iy}$ are jointly continuous with respect to $\topy_i$,
    then $ \quantTheory$ is topologisable by a second countable Hausdorff topology.
\end{theorem}

Recalling that $[0,1]$ is compact, thus the assumptions of \cref{thm:interpretation_hemicompact} are fulfilled.
\begin{lemma}
        \label{lem:standardAppl}
    For a language $\Sprache$ with the following signature in which the set $\actions$ of actions is countable, the theory map $\quantTheory$ is topologisable.
    \[
            \Sprache
        \Coloneqq
            {\wedge} \mid {\neg} \mid \top
            \mid \blank + r, r \in [0,1]
            \mid \blank - r, r \in [0,1]
            \mid \diamond_a, a \in \actions.
    \]
\end{lemma}


Now combining \cref{lem:standardAppl,thm:adequacy,thm:expressivity} we get that the modal language $\Sprache$ defined in \cref{lem:standardAppl} is both adequate and expressive for bisimulation pseudo-metrics $\bdistance$ defined in \cref{cor:bdistance-exists} for MDPs.
\begin{corollary}
        \label{cor:standardAppl}
    Let $\Sprache$ be the language as given in \cref{lem:standardAppl} and let $\gamma\colon \measSp\to\BfunctorMDP\measSp \in \ana$ be an MDP. Then the bisimulation pseudometric (defined in \cref{cor:bdistance-exists}) coincides with the logical distance $\ldistance$.

\end{corollary}


One may anticipate, following \cite{ChenClercPanaganden25}, to decompose the semantics of our diamond modality into two modalities: one modelling the expectation modality $\diamond_a' \varphi$ and the 0-ary reward modality $\reward_a$. The semantics of these two modalities given below in $\Sprache'$ is taken from \cite{ChenClercPanaganden25}. We argue next that this is unfortunately not possible without jeopardizing the adequacy result.


\begin{remark}
        \label{thm:Sprache_reward}
    Assume $c\in[0,1]$ and a language $\Sprache'$ with signature
    \begin{equation}
            \label{eq:Sprache_reward}
        \top
        \mid \neg \blank
        \mid \blank \land \blank
        \mid\reward_a, a \in \actions
        \mid\diamond'_a, a\in\actions
        \mid r\blank, r\in [0,1]
        \mid \blank + \blank
    \end{equation}
    and
    an interpretation depending on $\gamma\colon\measSp \to \BfunctorMDP\measSp \in \ana$ defined by:
    \begin{equation*}
        \interpret{\diamond_a \varphi}_\gamma(x)
    =   c \int \interpret{\varphi} \diff \gamma_{a,x}^0
     \qquad \text{and} \qquad
        \interpret{\reward_a} (x) =
        \gamma_{a,x}^1.
    \end{equation*}
Then the language $\Sprache'$  is expressive w.r.t. $\bdistance$ (i.e., $\bdistance \leq \mathbf{d}_{\Sprache'}$). However, $\Sprache'$ is not adequate since the binary (truncated) addition is not nonexpansive w.r.t. suprema distance.
\end{remark}


\section{Related work and concluding remarks}\label{sec:conc}

\subsection{Related work}
Our work is inspired by \cite{FernsPanangadenPrecup11} and establishes a quantitative version of Hennessy-Milner theorem for the therein defined bisimulation pseudometric.
To the best of our knowledge, such a generalisation is novel and has not been studied elsewhere in the literature. The key technical differences between the two works are as follows. First, our notion of conformance on continuous state MDPs is based on universal measurability; whilst, it is based on lower semi-continuity in \cite{FernsPanangadenPrecup11}. Note that every lower semi-continuous function is universally measurable. Second, our MDPs are coalgebras living in $\ana$ and the state space of an MDP is thus an analytic space in our paper; whilst, it is a Polish space in \cite{FernsPanangadenPrecup11}. Third, the bisimulation pseudometric $\bdistance$ defined in this paper is based on Wasserstein lifting; whilst, the bisimulation pseudometric (denoted $\bdistance_{FPP}$) of Ferns et al.\nolinebreak[3]\ is based on Kantorovich lifting. Note that the pseudometrics $\bdistance$ and $\bdistance_{FPP}$ are equivalent due to the Kantorovich-Rubinstein duality (\cref{lem:KantorovicRubinstein}). Finally, we employ the Kleene's fixpoint theorem to define $\bdistance$, whilst, Ferns et al.\nolinebreak[3]\ employed Banach fixed point theorem to define their bisimulation pseudometric $\bdistance_{FPP}$.

The recent work of Chen et al.\nolinebreak[3]\ \cite{ChenClercPanaganden25} on continuous time Markov processes (i.e. a family of $\BfunctorMP$-coalgebras indexed by non-negative real numbers) is also insightful, where a quantitative version of Hennessy-Milner theorem is also proven. The mathematical development followed in \cite{ChenClercPanaganden25} is worth comparing, especially when this family is restricted to a singleton coalgebra (say, for instance, $\gamma \colon \measSp \to \BfunctorMP\measSp$) and the $\sigma$-algebra $\sAlg$ is generated by a Polish topology on $X$. The functional $\mathcal F_c$ (for some $0<c<1$) defined in \cite{ChenClercPanaganden25} is not an endofunction in general on the lattice of lower semi-continuous functions on $\measSp$. Using the notations of this paper, $\mathcal F_c$ can be rewritten as:
$
\mathcal F_c (d) (x,y) = c \cdot \hat\sigma(d) (\gamma^0(x),\gamma^0(y))$ for every $x,y\in X$.

Nevertheless, to capture their bisimulation pseudometric (denoted $\bdistance_{CCP}$) by a fixpoint argument, the authors had to work with continuous distance functions on $X$. The usual Knaster-Tarski fixpoint theorem is inapplicable and the authors constructed $\bdistance_{CCP}$ as the limit of following pseudometrics $\delta_i$:
$
\delta_0=(\gamma^1\times\gamma^1)^* (d_E)$; $\delta_{i+1} = \mathcal F_c(\delta_i).$
As a result, the two bisimulation pseudometrics $\bdistance_{CCP}$ and $\bdistance$ (\cref{lem:dliftingForMP}) are different.

The recent works \cite{beohar_et_al:LIPIcs.STACS.2024.10,KupkeRot-CoindPred-2020,KomoridaKatsumataKu-ExpressivityofQuant,forster_et_al:LIPIcs.CSL.2023.22}  on developing expressive modal logic for a behavioural conformance that are defined by codensity lifting (called Kantorovich lifting in \cite{forster_et_al:LIPIcs.CSL.2023.22}) can, unfortunately, not be directly applied to the current setting. This is due to the underlying assumption of behavioural conformances defined internally in a complete lattice fibration (or equivalently using the language of topological functors \cite{forster_et_al:LIPIcs.CSL.2023.22}). To this end, we adopted a coupling-based lifting approach (inspired from \cite{BonchiKonigPetrisan18}) to define our bisimulation pseudometric. This adoption required significant effort in recasting old results from measure theory in our framework as outlined in Assumptions~\cref{ass:a1}-\cref{ass:a5}.

\subsection{Concluding remarks}
To summarise, we model both MRPs and MDPs with continuous state spaces as coalgebras in $\ana$ and define the notion of bisimulation pseudometric using the well known Kleene's fixpoint theorem. The latter was based on a given coalgebra $\gamma\colon \measSp \to \giry \measSp \in \ana$ and the fact that a functional $\gamma\circ \hat\sigma \colon \predicates \measSp \to \predicates\measSp$ is \ocpo-continuous (\cref{thm:WassersteinDist_cpoCont}), whose proof was in turn based on classical results from functional analysis. In addition, we also presented a `quantitative' modal logic $\Sprache$ whose formulae are interpreted as universally measurable predicates over the state space of an MDP and the logical distance $\ldistance$ generated by $\Sprache$ coincides with the bisimulation distance $\bdistance$. To prove the expressivity result (\cref{thm:expressivity}) is, certainly, more involved than the adequacy result (\cref{thm:adequacy}); nonetheless, they both require that the theory map $\quantTheory$ \eqref{eq:quantitativeTheory_def} is topologisable.

For future work the fact that in the expressivity result a topological structure on the formulas instead of any requirement on the statement was the key assumption may stipulate new perspectives.
A more concrete worthwhile enterprise would be to generalise the Stone-Weierstra\ss\ theorem for measurable spaces. This will help in directly invoking the argument to approximate a nonexpansive map $h$ by logical formulae in the proof of \cref{thm:expressivity}; thus, avoiding the topological arguments used here.



\bibliography{bibTeX,bibTeX_T2Aconverted}

\appendix

\section{Notations and background}
    \label{sec:notation_background}

The power set of a set $\spSet$ is denoted by $\powerSet (\spSet)$.
For a function $f\colon \spSet \to \spSet*$ we denote for a subset $A\subseteq \spSet$ its direct image under $f$ by $\dirImSet f (U)$ or $f[U]$ and for a $V \subseteq \spSet*$ the inverse image by $\invImSet f (V)$.

\subsection{Perfect measures}
    \label{ssec:perfectMeasures}
Perfect measures were introduced by Kolmogorov \cite[22--23]{GnedenkoKolmogorov49}.
The aim was to provide a convenient subclass of measures general enough for all applications.
There are many different equivalent definitions.
We choose the following one:
A measure space $\measdSp$ is called \definiendum{perfect}, if for any  separable metrisable space $(\spSet*, 	\topy) $ and every measurable map $f\colon \spSet\to\spSet*$ we have the following property: For every $A \in \sAlg $ and $ r < \meas(A) $ there is a compact set $K \subseteq \image f$ with $\meas( A \cap \invImSet f*K ) \geq r$, cf.\nolinebreak[3]\ \cite[451O(a)]{Fremlin}.
Also $\meas$ is called \definiendum{perfect} in this case.
As a direct consequence of the definition, perfectness is functorial, i.e.\ the push-forward of a perfect measure is perfect.

In typical real-world applications two points of a measurable space should only be distinguished by the \salg{} if they can be distinguished by an observation. 
As only a finite amount of observations with limited precision can be made per unit of time, there should be a countable subset $\paving S \subseteq \sAlg$ distinguishing as strong a $\sAlg$ does ($ \forall x,y \in \spSet \forall A \in \sAlg\colon x \in A \wedge y \notin A \implies (\exists S \in \paving S\colon \# S \cap \{x,y\} = 1  )$).
A measurable space enjoying this property is called \definiendum{countably fibered}.
Perfect countably fibered probability spaces are actually---despite being a very general notion, including, e.g.\nolinebreak[3]\ analytic spaces---quite close to standard spaces \cite[\S~8 Rem.]{Faden85}:
Any such space is almost pre-standard with respect to some sub-\salg{} $\sAlg' \subseteq \sAlg$.
Almost pre-standard means that a space is standard when restricted to a Borel set of full measure and identifying all point not distinguished by $\sAlg$ \cite{Faden85}.
For countably generated spaces perfectness can even be characterised equivalently by being almost pre-standard \cite[Thm.~6]{Faden85}.

\subsection{\SuslinText{} operation and smooth spaces}
    \label{ssec:SouslinOp_smoothSp}
\begin{subequations}
For a function $ f_{(\blank)}\colon \omega \to X $
let $ f_{(\blank)}|_{\leq i} $ denote the restriction to the first $k$ indices.
Further let $\omega^{<\omega}$ denote the set of all finite sequences in $\omega$.
The \SuslinText{} operation, cf.\nolinebreak[3]\ \cite[25.4]{Kechris12} or \cite[421B]{Fremlin}, is denoted by $\Souslin$;
we recall that it is defined by
\begin{equation}
        \Souslin \paving P
    =   \bigcup_{n_{(\blank)} \in \omega^\omega}
            \bigcap_{i \in \omega} A_{n_{(\blank)}|_{\leq i}}
\end{equation}
for a \SuslinText{} scheme $A_{(\blank)}\colon \omega^{<\omega}$.
Further, we remind of its elementary properties: that for any paving $\paving P$ and countable family $\paving S \subseteq \Souslin \paving P$
\begin{align}
        \label{eq:Souslin_countableUnion}
    \bigcup \paving S &\in \Souslin \paving P
\\
        \label{eq:Souslin_countableIntersection}
    \bigcap \paving S &\in \Souslin \paving P
\end{align}
\cite[421E]{Fremlin}, also for any function $f\colon \spSet \to \spSet*$ and paving $\paving Q$ on $\spSet*$ that
\begin{equation}
        \label{eq:Souslin_pullback}
    \invImSet f*[\Souslin \paving Q] = \Souslin( \invImSet f*[\paving Q] )
\end{equation}
\cite[421Cc]{Fremlin}, as well as monotonicity \cite[421Ca]{Fremlin} and idempotence \cite[421D]{Fremlin}, that for any paving $\paving P$
\begin{align}
        \label{eq:Souslin_monotone}
    \paving P &\subseteq \Souslin \paving P
\\
        \label{eq:Souslin_idempotent}
    \Souslin (\Souslin \paving P) &= \Souslin \paving P
\text.
\end{align}
\end{subequations}

Denote the Giry monad by $
    \giry\measSp = (\mathrm{M}_{\measSp}, \mathrm{A}_{\measSp})
$.

We also recall a generalisation of analytic spaces: smooth spaces as introduced by Falkner \cite{Falkner81}.
It can be defined as follows \cite[1.3]{Falkner81}:
\begin{definition}
        \label{def:smooth}
    A measurable space $\measSp$ is called \definiendum{smooth} if 
        for any measurable space $\measSp$ and any $A \in \Souslin(\sAlg*\measTimes\sAlg)$ the projection on the first component $\proj{\spSet*} A$ is in $\Souslin{\sAlg*}$.
\end{definition}

Note the following fact \cite[1.3]{Falkner81}:
\begin{lemma}
        \label{lem:analytic_=>_smooth}
    Every analytic space is smooth.
\end{lemma} 


\section{Proof of lemmas and theorems from \cref{sec:MDP}}
\begin{proposition}
        \label{prop:pullbackSet_measurable}
    Let $ 
        \measSp \xrightarrow{f} 
        \measSp*
        \xleftarrow{f'} \measSp[']*
    $  be a cospan of measurable spaces and assume that $\measSp*$ is countably separated.
    Then $L \coloneqq \spSet \setTimes[f][f'] \spSet* \coloneqq \setBuilder{(x,x') \in \spSet\times\spSet*}{f(x) = f'(x')} \in \sAlg \otimes \sAlg' $.
\end{proposition}
\begin{proof}
    Let $\paving S  $ denote a countable separating family on $\sAlg*$ such that $\complOp S \in \paving S $ for any $S \in \paving S $.
    Observe that
    \begin{align*}
            L
        &=  \setBuilder{(x, x') \in \spSet  \times \spSet' }{f(x) = f'(x')}
    \\  &=  \setBuilder
                {(x, x') \in \spSet  \times \spSet' }
                {\forall\, S \in \paving S \colon f(x) \in S \implies f'(x') \in S}
    \\  &=  \bigcap_{S \in \paving S }
            \setBuilder
                {(x, x') \in \spSet  \times \spSet' }
                { f(x) \in S \implies f'(x') \in S}
    \\  &=  \bigcap_{S \in \paving S }
            \invImSet f*(\complOp S) \times \spSet'  \cup \spSet  \times \invImSet f*(S)
    \\  &\in\sAlg\otimes\sAlg'
    \text.
    \qedhere
    \end{align*}
\end{proof}

\begin{corollary}
        \label{cor:graph_measurable}
    For a measurable map $f\colon \measSp \to \measSp*$ from a measurable space to a countably separated measurable space, the graph $\graph f$ is in $\sAlg \otimes \sAlg*$.
\end{corollary}
\begin{proof}
    Apply \cref{prop:pullbackSet_measurable} to $ \measSp* \xrightarrow{\id} \measSp* \xleftarrow{f} \measSp $.
\end{proof}

\subsection{Proof of \cref{lem:directImage}}

\begin{proof}
    The existence of countable meets and joints follows directly from the fact that the $\Souslin \sAlg$ (even for arbitrary subsets $\sAlg \subseteq \powerSet \spSet$) is closed under countable union and intersection \cref{eq:Souslin_countableUnion,eq:Souslin_countableIntersection}, e.g.\ for meets $ \invImSet*{\inf_{i\in \omega} P_i}([0,r]) = \bigcup_{i \in \omega} ( \invImSet{P_i} [0,r]) \in \Souslin \sAlg $ for any $ r \in [0,1] $ assuming that all $ P_i \in \spredicates \measSp $.

    For the direct image let us first verify that $\inf(f)$ is well-defined, i.e.\ hit its codomain.
    Take any $r\in [0,1]$ and observe
    \begin{align*}
            \invImSet*{\exists_f P}( \intervOpen  r )
        &=  \setBuilder y { (\inf(f) P) (y) \in \intervOpen  r}
    \\  &=  
            \setBuilder*y { \inf P[\invImSet f*(\{y\})] \in \intervOpen  r }
    \\  &=  \setBuilder y {\exists s \in P[\invImSet f*(\{y\})] \colon s \in \intervOpen  r }
    \\  &=  \setBuilder y {\exists x\colon f(x) = y \wedge P(x) \in \intervOpen r }
    \\  &=  \proj Y \left(\graph{f} \cap \invImSet P*(\intervOpen r ) \times \spSet*\right)
    \end{align*}
    Note that $ \in \Souslin(\sAlg \measTimes \sAlg*) $ due to \cref{cor:graph_measurable,eq:Souslin_countableIntersection}.
    Thus the projection of this set to $Y$ is in $\Souslin \sAlg*$ as $\measSp$ is smooth \cref{def:smooth}.

    As for the adjunctions we have to check
    \[
    \begin{prooftree}
        \hypo{P \preceq Q\circ f }
        \infer[double]1
            {\inf(f) P \preceq Q }
    \end{prooftree}
    \]
    for all $P \in \predicates^\udarrows  \measSp$ and $Q \in \predicates^\udarrows  \measSp*$.
    But these facts are known from the set case.
\end{proof}

\section{Proof of lemmas and theorems from \cref{sec:bdistance}}
    \label{ssec:predicate_lifting}

\subsection{Proof of \cref{thm:predLifting_semimeas}}

For $ r \in [0,1] $ let $ 
    \intervLowerOpen r = \setBuilder{s \in [0,1]}{s < r}
$, the \definiendum{lower interval}, and $ 
    \intervUpperOpen r = \setBuilder{s \in [0,1]}{s > r}
$, the \definiendum{upper interval}.
Also, let $
    \intervLowerCl r = \setBuilder{s \in [0,1]}{s \leq r}
$ and $
    \intervUpperCl r = \setBuilder{s \in [0,1]}{s \geq r}
$ the corresponding \definiendum{closed lower\textnormal /upper}.
Our next challenge is to prove that predicate lifting preserves semi-measurability, which reduces to a question whether the following generators are \SuslinText{} sets.
\begin{equation}
        \label{eq:giry_sAlg_generators}
    \bar E^A_r \coloneqq \invImSet{\eval{A}}* \intervLowerCl r 
\text{,}\quad
    E_A^r \coloneqq \invImSet{\eval{A}}* \intervUpperOpen r 
\end{equation}
(with $\eval{A}$ interpreted with respect to the outer measure, if $A$ not measurable, i.e.\nolinebreak[3]\ $\eval{A}(\meas) = \meas^*(A)$)
with $ A \subseteq \spSet$ and $r \in [0,1]$.
The fact $ \bar E^A_r \in \Souslin \giry\measSp $ is proved using an idea from \cite[2.6.2]{Doberkat14} but presented only in terms of measurable spaces.

\begin{lemma}
        \label{lem:countGen_closedGenPav}
    Let $\measSp$ be a countably generated measurable space.
    Then there is a countable subpaving $\paving S \subseteq \sAlg$ generating $ \sAlg $ such that for
    \begin{equation}
        \paving F = \setBuilder*{\textstyle \decrIntersecLim_n S_n}{ S_1, \ldots \in \paving S, S_1 \supseteq S_2 \supseteq \ldots, n \in \omega }   
    \end{equation}
    we have
    \begin{enumerate}
        \item\label{lst:lem:countGen_closedGenPav_bottom}
            $\emptyset \in \paving F$,
        \item\label{lst:lem:countGen_closedGenPav_finUnion} 
            $ \paving F $ is closed under finite unions,
        \item\label{lst:lem:countGen_closedGenPav_countIntersec} 
            $ \paving F $ is closed under countable intersections,
        \item\label{lst:lem:countGen_closedGenPav_compl} 
            $ \paving S $ is closed under complement.
    \end{enumerate}
\end{lemma}
\begin{proof}
    Let $\paving S$ be a countably generating set.
    We may assume that $\paving S $ is an algebra.
    Thus \cref{lst:lem:countGen_closedGenPav_compl,lst:lem:countGen_closedGenPav_bottom} hold.
    Further set $\paving S_\delta = \setBuilder{\bigcap_n S_n}{ S_1, \ldots \in \paving S, n \in \omega } $.
    We check that $ \paving S_\delta $ is closed under finite unions and countable intersection:
    For finite unions this follows from distributivity, i.e.\nolinebreak[3]\ $ 
                \bigcap_i A_i \cup \bigcap_j B_j 
            =   \bigcap_j (\bigcap_i A_i \cup  B_j)
            =   \bigcap_{i, j} A_i \cup  B_j
        $.
    Countable intersection is simply reindexing (actually, this argument works for arbitrary intersections considering that $\paving S$ is only countable).
    Finally, we conclude the proof by showing $\paving S_\delta = \paving F$.
    Obviously, $ \paving S_\delta \supseteq \paving F $.
    For the other inclusion note that $\paving S$ is closed under finite intersections.
\end{proof}

\begin{lemma}
        \label{lem:countGen_innerReg}
    Let $\measSp$ be a measurable space and $\paving S \subseteq \sAlg$  a generating paving subject to \cref{lst:lem:countGen_closedGenPav_bottom,lst:lem:countGen_closedGenPav_compl,lst:lem:countGen_closedGenPav_countIntersec,lst:lem:countGen_closedGenPav_finUnion} of \cref{lem:countGen_closedGenPav} (with $\paving F$ as defined therein).
    Then any finite measure on $\measSp$ is $\paving F$-inner regular.
\end{lemma}
\begin{proof}
    \cite[412C]{Fremlin}.
\end{proof}

\begin{subequations}
For a subpaving $\paving F \subseteq \sAlg$ as provided by \cref{lem:countGen_closedGenPav} set
\begin{align}
        \label{eq:F_wittnesses}
        \mathfrak F
    &=  \setBuilder{ S_{(\blank)}\colon \omega \to \paving S}{\text{decreasing, } S_0 = \spSet}
\text,
\\
        \label{eq:FA_wittnesses}
        \mathfrak F_A
    &=  \setBuilder*
            { S_{(\blank)} \in \mathfrak F }
            {\textstyle A \supseteq \bigcap_{i\in \omega} S_i }
&&\text{for any } A \subseteq X \text{ and}
\\
        \label{eq:FA_wittnesses_finite}
        \mathfrak K_A
    &=  \setBuilder*{ S_{(\blank)} \in \mathfrak F_A }{ \setBuilder{S_i}{i\in\omega} \text{ finite} }
\text.
\end{align}
\end{subequations}
Note that $\mathfrak K_A$ is countable as $\paving S$ is.

\begin{lemma}
        \label{lem:generatorMeasurable_embedding}
    There exists an map
    \begin{equation}
            \Phi
        \colon 
                \complOp{\mathfrak K_A} = \mathfrak F_A \setminus \mathfrak K_A
            \to \omega^\omega
    \end{equation}
    such that
    \begin{enumerate}
        \item \label{lst:lem:generatorMeasurable_embedding_inj}
            $\Phi$ is injective;
        \item \label{lst:lem:generatorMeasurable_embedding_injAtIndex}
            for all $ S_{(\blank)}, T_{(\blank)} \in \complOp{\mathfrak K_A} $ and $ i \in \omega $:
            $ S|_{\leq i} = T|_{\leq i} $ implies $ \Phi(S_{(\blank)})|_{\leq i} = \Phi(T_{(\blank)})|_{\leq i} $;
        \item \label{lst:lem:generatorMeasurable_embedding_inj_reformulated}
            for all $ S_{(\blank)}, T_{(\blank)} \in \complOp{\mathfrak K_A} $:
            $ S_{(\blank)} = T_{(\blank)} \iff \forall k \in \omega\colon \Phi(S_{(\blank)}) = \Phi(T_{(\blank)}) $;
    \end{enumerate}
    ---let $ 
            F_i
        =   \setBuilder[\big]
                { n_{(\blank)}|_{\leq i} }
                { n_{(\blank)} \in \Phi[\complOp{\mathfrak K_A}] } 
    $ be the set of initial pieces of length $i\in \omega$---
    \begin{enumerate}
    \setcounter{enumi} 3
        \item  \label{lst:lem:generatorMeasurable_embedding_inj_level_i} 
            for each $i \in \omega$ all witnesses $S_{(\blank)}$ of $ n_{(\blank)}|_{\leq i} \in F_i $ coincide on the first $i+1$ coordinates, i.e.\nolinebreak[3]\ $ S_0, \ldots, S_i $ are determined for every $S_{(\blank)}$ with $\Phi(S_{(\blank)}) \in F_i $
        \item  \label{lst:lem:generatorMeasurable_embedding_imageCompl}
            $ \complOp{\mathfrak K_A} = \bigcap_{i\in\omega} \setBuilder{S_{(\blank)}}{\Phi(S_{(\blank)})|_{\leq i} \in F_i } $.
    \end{enumerate}
\end{lemma}
\begin{proof}[Proof (using countable choice)]
    First, build a tree $ \mathfrak T $ labelled in $ \paving S $ as follows:
    View $\paving S$ as a tree with $\spSet$ as root.
    Remove all paths $S_1 \supseteq S_2 \supseteq \ldots $ with $ A \nsupseteq \bigcap_i S_i$.
    
    Let $\mathcal T_0$ be a singleton set and set  $ S_0 = \spSet $.
    For any nested $S_0 \supseteq \ldots \supseteq S_i$ let $\mathcal{T}_{i+1, S_0, \ldots, S_i}$ denote the set of all offsprings of $S_i$ in $\mathfrak T$ and choose an injective function
    \[
        \Phi_{i+1, S_0, \ldots, S_i} \colon \mathcal{T}_{i+1, S_0, \ldots, S_i} \to \mathbb N
    \]
    ---but only depending on the set $\mathcal{T}_{i+1, S_0, \ldots, S_i}$.
    Define $\Phi$ by the following ``trace like'' formula
    \[
            \Phi(S_{(\blank)})
        =   ( \Phi_{i, S_0, \ldots, S_{i-1}} (S_i) )_{i \in \omega}
    \text.
    \]

    For injectivity, \cref{lst:lem:generatorMeasurable_embedding_inj}, assume that $ \Phi(S_{(\blank)}) = \Phi(T_{(\blank)}) $ and conclude by induction that $ S_{(\blank)} = T_{(\blank)} $:
    Obviously, $ \Phi(S_{(\blank)})_0 = \Phi(T_{(\blank)})_0 $, the image of the singleton $\mathcal T_0$.
    For the step assume $ S_{(\blank)}|_{\leq i} = T_{(\blank)}|_{\leq i} $.
    Then $ \Phi_{i+1, S_0, \ldots, S_i} = \Phi_{i+1, T_0, \ldots, T_i}$ is injective.
    Thus $S_{i+1} = T_{i+1}$.
    \Cref{lst:lem:generatorMeasurable_embedding_inj_reformulated} is just a reformulation of \cref{lst:lem:generatorMeasurable_embedding_inj}.
    The claim \cref{lst:lem:generatorMeasurable_embedding_injAtIndex} follows from the fact that $\Phi_{i+1, S_0, \ldots, S_i}$ depends only on $\mathcal{T}_{i+1, S_0, \ldots, S_i}$.

    Claim~\ref{lst:lem:generatorMeasurable_embedding_inj_level_i} follows from induction.
    The base case is again trivial as always $S_0 = \spSet$.
    For the induction step assume for all witnesses $ (S_0, \ldots, S_i, \ldots) $ of $ n_{(\blank)}|_{\leq i+1} \in F_{i+1} $ (i.e.\nolinebreak[3]\ $\Phi(S_{(\blank)}) \in F_{i+1}$) coincide.
    So $S_{i+2}$ is due to injectivity uniquely determined by $ \Phi_{i+1,S_0,\ldots, S_i} $.

    The final claim, \cref{lst:lem:generatorMeasurable_embedding_imageCompl},---applying $\Phi$ to both sides and doing a reformulation in terms of quantifiers---becomes $ 
                    n_{(\blank)} \in \Phi[\complOp{\mathfrak K_A}] 
            \iff    
                    \forall i \in \omega\colon n_{(\blank)}|_{\leq i} \in F_i $.
    Here the direction ``$\implies$'' is trivial.
    For the implication ``$\impliedby$'', assuming $\forall i \in \omega\colon n_{(\blank)}|_{\leq i} \in F_i$, 
    a witness of $n_{(\blank)} \in \Phi[\complOp{\mathfrak K_A}] $ can obviously be constructed using \cref{lst:lem:generatorMeasurable_embedding_inj_level_i}.
\end{proof}

\begin{proposition}
        \label{prop:giry_sAlg_generators_measurable}
    Let $\measSp$ be a countably generated measurable space. 
    The generators $\bar E^A_r$ (for any subset $A \subseteq \spSet$) 
    (cf.\nolinebreak[3]\ \cref{eq:giry_sAlg_generators})
    are \SuslinText{} sets.
\end{proposition}

\begin{proof}
    Let $\paving S \subseteq \sAlg$  a generating paving subject to \cref{lst:lem:countGen_closedGenPav_bottom,lst:lem:countGen_closedGenPav_compl,lst:lem:countGen_closedGenPav_countIntersec,lst:lem:countGen_closedGenPav_finUnion} of \cref{lem:countGen_closedGenPav} (with $\paving F$ as defined therein).
    Note that
    \begin{align*}
            \bar E^A_r \qquad
        &\refrel{eq:giry_sAlg_generators}=  
            \setBuilder{\meas\in \giry\measSp}
                { \meas_*(\complOp A) > 1 - r }
    \\
        &\stackrel{\mathllap{\substack{
              \text{\cref{lem:countGen_innerReg},} 
            \\ \text{\cref{eq:F_wittnesses,eq:FA_wittnesses}} }}}=   
            \bigcup_{ S_{(\blank)} \in \mathfrak F_{\complOp A} }
                \setBuilder*
                    {\meas\in \giry\measSp}
                    {\textstyle \meas \left( \decrIntersecLim_{i\in\omega} S_i \right) > 1 - r }
    \\
        &\refrel{eq:giry_sAlg_generators,eq:FA_wittnesses_finite}=       
                \bigcup_{ S_{(\blank)} \in \mathfrak K_{\complOp A} } 
                    \decrIntersecLim_{i\in\omega} E_{ S_i }^r 
            \cup
                \bigcup_{ S_{(\blank)} \in \complOp{\mathfrak K_{\complOp A}} }
                    \decrIntersecLim_{i\in\omega} E_{ S_i }^r 
    \text.
    \end{align*}
    As $\mathfrak K_{\complOp A}$ is countable, we have $ 
            \bigcup_{ S_{(\blank)} \in \mathfrak K_{\complOp A} } 
                \decrIntersecLim_{i\in\omega}  E_{ S_i }^r 
        \in 
            \mathrm{A}_{\measSp}
    $.
    For the other set observe that by \cref{lem:generatorMeasurable_embedding}
    \begin{align*}
            \bigcup_{ S_{(\blank)} \in \complOp{\mathfrak K_{\complOp A}} }
                    \decrIntersecLim_{i\in\omega}  E_{ S_i }^r 
        &=
            \bigcup_{n_{(\blank)} \in \omega^\omega }
            \begin{cases*}
                \decrIntersecLim_{i\in\omega}  E_{ S_i }^r 
                    & the $ S_{(\blank)} $ with $ \Phi( S_{(\blank)} ) = n_{(\blank)} $,
            \\
                \emptyset   & else.
            \end{cases*}
    \intertext{As $ S_0, \ldots, S_i $ depend only on $n_{(\blank)}|_{\leq i}$ by \cref{lst:lem:generatorMeasurable_embedding_inj_level_i} of \cref{lem:generatorMeasurable_embedding}, 
    $  E_{ \invImSet\Phi* ( \{ n_i \} ) (i)}^r  $ is actually a \SuslinText{} scheme.
    Thus}
        &\in \Souslin\mathrm{A}_{\measSp}
    \qedhere
    \end{align*}
\end{proof}

\begin{lemma}
        \label{lem:decrFct_are_RiemannInt_disc}
    The Lebesgue integral of a non-increasing function $f\colon [0,1] \to [0,1]$ can be approximated with intervals of the form $[0,r]$ with $r \in \mathbb Q$, i.e.\
    \[
        \int f
        =   \inf \setBuilder*
                {\sum_{i=1}^k (s_i - s_{i-1}) r_i}
                { \begin{multlined}\textstyle
                    (s_i, r_i) \in (([0,1]\cap \mathbb Q)^{\times 2})^{{} \times k} 
                    \textstyle \text{ with } k \in \omega\text,
                    \\\textstyle  s_1 < \ldots < s_k \text{ and }
                    \sum_{i=1}^k r_i \chi_{(s_{i-1},s_i)} \geq f 
                \end{multlined}}
    \text.
    \]
\end{lemma}

    This statement is basically a reformulation of Riemann integrability.
\begin{proof}
    First note that $f$ is continuous at all but countably many points:
    At any discontinuity point $t$ of $f$ the difference $
        \delta_t = \lim_{s\nearrow t} f(s) - \lim_{s\searrow t} f(s)
    $ must be positive.
    But $\sum \setBuilder{\delta_t}{ t \in [0,1] } \leq 1 $.

    Thus $f$ is almost everywhere continuous with respect to Lebesgue measure.
    Hence it is Riemann integrable \cite[134L]{Fremlin}.
    This implies that $\int f$ can be approximated from below up to $\varepsilon$ via a dissection $ 0 \leq s_1 \leq \ldots \leq s_k$ by $
        \sum_{i=1}^k (s_k - s_{k-1}) \sup \setBuilder{ f(s) }{ s \in (s_{i-1}, s_i) }
    $ (note $k = k(\varepsilon)$).
    Now choose $\varepsilon \in 2^{-\omega}$ set $
            s_{i,\varepsilon,k}
        =   \frac{\varepsilon}{k} \left\lceil \frac{k}{\varepsilon} s_i \right \rceil 
    $ and $
            r_{i,\varepsilon}
        =   \varepsilon \left\lceil 
                \sup \setBuilder{ f(s) }{ s \in (s_{i-1}, s_i) } / \varepsilon 
            \right\rceil
    $.
    Then 
    \begin{align*}
        \MoveEqLeft
            \sum_{i=1}^k 
                (s_{i,\varepsilon,k} - s_{i-1,\varepsilon,k}) r_{i,\varepsilon}
            - \sum_{i=1}^k (s_k - s_{k-1}) \sup_{s \in (s_{i-1}, s_i)} f(s)
    \\  &=
            \sum_{i=1}^k 
            \begin{multlined}[t]
                (s_{i,\varepsilon,k} - s_{i-1,\varepsilon,k}) r_{i,\varepsilon}
                \\ 
            - (s_k - s_{k-1}) r_{i,\varepsilon}
            \end{multlined}
            + \sum_{i=1}^k (s_k - s_{k-1}) \left(\sup_{s \in (s_{i-1}, s_i)} f(s) - r_{i,\varepsilon} \right)
    \\  &\leq
            \sum_{i=1}^k s_{i,\varepsilon,k} - s_k + 1\cdot \varepsilon
    \\  &\leq   k \cdot \nicefrac\varepsilon k + \varepsilon
    \text.
    \end{align*}
    As $\varepsilon > 0$ was chosen arbitrarily, we can approximate $\int f$ from above by $ \sum_{i=1}^k 
                (s_{i,\varepsilon,k} - s_{i-1,\varepsilon,k}) r_{i,\varepsilon} $.
\end{proof}

\begin{corollary}
        \label{cor:predLifting_semimeas}
The mapping $\sigma_{(\blank)}$ defined in \eqref{eq:ExpectationLifting} is a natural transformation; thus, an indexed category morphism. Moreover, $\sigma$ preserves directed suprema.
\end{corollary}

\begin{proof}
    To verify that $\sigma_{\measSp}$ is well-defined, we must show lower semimeasurability of $ \sigma_{\measSp}(P) $ for any  $ P \in \predicates\measSp $.
    To check this take some $ r \in [0,1] $ and observe
    \begin{align*}
        \MoveEqLeft
            \invImSet{\lambda_{\measSp}(P)} \intervOpen r
    \\
        &=  \setBuilder*{\meas}{ \textstyle\int P \diff \meas < r }
    \intertext{using Choquet integration \cite[1.61]{Doberkat07} this can be expressed as}
        &=  \setBuilder*
                {\meas}
                { \textstyle\int_0^1 \meas( \invImSet P* (\intervUpperOpen t) ) \diff t < r }
    \intertext{here we can approximate the integral $\int_0^1$ over a decreasing function from above using \cref{lem:decrFct_are_RiemannInt_disc}}
        &=  \bigcup
            \setBuilder*{
                \setBuilder*
                    {\meas}
                    { \begin{multlined}
                    \textstyle \sum_{i=1}^k r_i \chi_{(s_{i-1},s_i)} 
                    \\ \textstyle
                        \geq \lambda t.\meas( \invImSet P*(\intervUpperOpen t) ) 
                    \end{multlined} }
            }{ \begin{multlined}
                \textstyle (s_i, r_i) \in (([0,1]\cap \mathbb Q)^{\times 2})^{{} \times k} 
                \\ \textstyle\text{ with } k \in \omega\text, s_1 < \ldots < s_k 
                \\ \textstyle\text{ and }
                    \sum_{i=1}^k (s_i - s_{i-1}) r_i < r
            \end{multlined} }
    \end{align*}
    Now note that 
    \begin{align*}
            &\textstyle\sum_{i=1}^k r_i \chi_{(s_{i-1},s_i)} \geq \lambda t.\meas( \invImSet P*(\intervUpperOpen t) )
    \\  \iff&
            \forall i = 1, \ldots, k\colon
                \meas( \invImSet P*(\intervUpperOpen t) ) \geq r_i
            \text{ for } t \in [0,r_i)
    \\  \iff&
            \forall i = 1, \ldots, k\colon
                \meas \in \invImSet{\eval{\complOp{\invImSet P*(\intervUpperOpen t)}}}* \intervLowerOpen{1-r_i}
            \text{ for } t \in [0,r_i)
    \intertext{as $
            \invImSet{\eval{\complOp{\invImSet P*(\intervUpperOpen t)}}}* \intervLowerOpen{1-r_i}
        =   \incrUnionLim_{s\in [0,1-r_i) \cap \mathbb Q}   
                \invImSet{\eval{\complOp{\invImSet P*(\intervUpperOpen t)}}}* \intervLowerCl{s}
        \stackrel{\text{\cref{eq:giry_sAlg_generators}}}=
            \incrUnionLim_{s\in [0,1-r_i) \cap \mathbb Q} 
            \bar E_s^{\complOp{\invImSet P*(\intervUpperOpen t)}}
    $}
        \iff&
            \meas \in 
                \bigcap_{i = 1, \ldots, k}
                \adjustlimits
                \incrUnionLim_{ t \in [0,r_i) \cap \mathbb Q }
                \incrUnionLim_{ s\in [0,1-r_i) \cap \mathbb Q }
                    \bar E_s^{\invImSet P*(\intervUpperOpen t)}
    \end{align*}
    As $
            \bar E_s^{\invImSet P*(\intervUpperOpen t)}
        \stackrel{\text{\cref{prop:giry_sAlg_generators_measurable}}}\in 
            \Souslin(\pavGen\sigma{\invImSet P*(\intervUpperOpen t)}[t \in [0,1]\cap Q])
        \stackrel{\Souslin \text{ monotone}}\subseteq
            \Souslin\mathrm{A}_{\measSp}
    $, we, finally, obtain $ \invImSet{\lambda_{\measSp}(P)} \intervOpen r \in \Souslin\mathrm{A}_{\measSp} $.

    Predicate lifting is a lattice morphism: Note that $ P \leq P' $ implies that for any measure $\meas$ on $\measSp$ we have $
        \int P \diff \meas \leq \int P' \diff \meas
    $ \cite[122Od]{Fremlin}.
    It remains to check that countable directed suprema are preserved.
    This follows immediately from Levi's theorem, also called monotone convergence theorem, \cite[123A]{Fremlin}:
    Take some $\omega$-chain $ P_0 \geq P_1 \geq \ldots \in \predicates\measSp $ and observe for any $\meas \in \giry\measSp$ that
    \begin{align*}
            (\inf_i \lambda_{\measSp} (P_i))(\meas) 
        &=   \inf_i (\lambda_{\measSp} (P_i)(\meas))
    \\  &=   \inf_i \int P_i \diff \meas
    \\  &=   1 - \sup_i \int 1 - P_i \diff \meas
    \\  &\citerel[123A]{Fremlin}=
            1 -  \int \sup_i 1 - P_i \diff \meas
    \\  &=   \int \inf_i P_i \diff \meas
    \\  &=   \lambda_{\measSp}( \inf_i P_i )
    \text.
    \end{align*}

For the natural transformation property note that given $ f\colon \measSp \to \measSp* $ and a predicate $Q \in \predicates\measSp*$ we have $\int Q \comp f \diff \meas = \int Q \diff f_* \meas$ for any measure $\meas$ on $\measSp$.
\end{proof}

\begin{corollary}
        \label{cor:predLifting_semiMeasDouble}
    Predicate lifting given in \eqref{eq:ExpectationLifting} transforms a universally measurable predicate into a Borel predicate.
\end{corollary}

We denote by $ \mathsf{usPred}\measSp $ the set of \definiendum{upper semi-measurable predicates}, i.e.\ functions $P\colon \measSp \to [0,1]$ with $ \invImSet f([r,1]) \in \Souslin(\sAlg) $ for all $r\in [0,1]$.

\begin{proof}
    We already know that for any $P \in \predicates\measSp $ we have $ P \in \operatorname{\mathsf{usPred}}\measSp $ by \cref{cor:predLifting_semimeas}.
    So it remains to verify $ P \in \spredicates $.
    Again by \cref{cor:predLifting_semimeas} we have $ 
        \lambda_{\measSp}(1 - P) \in \operatorname{\mathsf{usPred}}\measSp
    $.
    Thus $ 
        1 - \lambda_{\measSp}(1 - P) \in \spredicates\measSp
    $.
    But $
            1 - \lambda_{\measSp}(1 - P)
        =   1 - (1 - \lambda_{\measSp}(P))
        =   \lambda_{\measSp}(P)
    $ for any $\meas \in \giry\measSp$.

    The last claim follows from Lusin's separation theorem \cite[422J]{Fremlin}.
\end{proof}

\section{Proofs of lemmas and theorems from \cref{sec:MDP}}
\begin{arXivVersion}
\subsection{Proof of \cref{thm:WassersteinDist_cpoCont}}    
\end{arXivVersion}
The proof of \cref{thm:WassersteinDist_cpoCont} is based on Sion's minmax theorem \cite[Thm.~3]{Simons95} which we recall first.

\begin{lemma}
        \label{lem:minimax_Sion}
    \begin{subequations}
    Let $U$ be a convex subset of a linear topological space,
    $V$ a compact convex subset of a linear topological space, and
    $ f\colon U \times V \to \mathbb R $ be upper semicontinuous on $U$ and lower semicontinuous on $V$.
    Suppose that
    \begin{align}
        \forall y \in V, \lambda \in \mathbb R\colon
        &  &
        \setBuilder{x \in U}{ f(x,y) \geq \lambda }
        & \text{ is convex and}
    \\
        \forall x \in U, \lambda \in \mathbb R\colon
        &  &
        \setBuilder{y \in V}{ f(x,y) \leq \lambda }
        & \text{ is convex}
    \text.
    \end{align}
    \end{subequations}
    Then $\displaystyle
            \adjustlimits\inf_{y\in V}\sup_{x \in U} f(x,y)
        =   \adjustlimits\sup_{x \in U}\inf_{y\in V} f(x,y)
    $.
\end{lemma}

In addition, we also need some important results from functional analysis; in particular, non-topological version of Riesz–Markov–Kakutani representation theorem \cite[IV.5.1]{DunfordSchwartz58} and Banach-Alaoglu theorem \cite[V.4.2]{DunfordSchwartz58}.

The vector space $\unifFct (X)$ of bounded functions $ f\colon X \to \mathbb R $ is endowed with the uniform norm given by
\[
        \Norm{f}{u}
    =
        \sup\setBuilder{\abs{f(x)}}{x \in \spSet}
\text.
\]
It is well-known that this space $\unifFct (X)$ is complete \cite[IV.5]{DunfordSchwartz58}.

Let $ \unifFct\measSp $ denote the Banach space \cite[IV.5]{DunfordSchwartz58} consisting of limits of simple functions on $\measSp$ with respect to the uniform norm $\Norm{\blank}{u}$.
It is also functorial.
Moreover, we recall that $B^*$ denotes the dual of a Banach space $B$, i.e.\ all bounded linear functionals $B \to \mathbb R$.
This dual can be either endowed with the uniform norm obtaining a Banach space or with the weak-* topology.
The assignment $(\blank)^*$ is functorial with respect to both choices. We write the dual vector space of $\unifFct\measSp$ simply as $\unifFct^*\measSp$.

Recall that a \definiendum{charge} is the same thing as a measure but only finitely additive.
Denote by $ \charges\measSp $ the set of all charges on $\measSp$.
We may also view $\charges$ as a functor $\Cat{Meas} \to \Cat{Meas}$ by endowing $\charges\measSp$ with the same \salg{} as $\giry\measSp$, i.e.\ the one making all $\eval{A}$ with $A \in \sAlg$ measurable. 
There is the following classical duality between charges and positive functionals given by the isometric (with respect to the uniform norm) isomorphism \cite[IV.5.1]{DunfordSchwartz58}
\begin{align}
        \label{eq:duality_charges_posFunctionals}
        \operatorname{Int} \colon \charges\measSp \to \unifFct^*\measSp
    \text,
    \quad
        \meas \mapsto \lambdCalc f.\, \int f \diff \meas
\text.
\end{align}
The above duality given by $\operatorname{Int}$ can be viewed as a non-topological version of the Riesz-Markov-Kakutani representation theorem. Note that probability charges $\meas$ are characterised by being positive, i.e.\ $
    \int f \diff \meas = \operatorname{Int}(\meas)(f) \geq 0
$ for any $f \in \unifFct\measSp$ with $f\geq 0$, and normed, i.e.\ $\int 1 \diff \meas = \operatorname{Int}(\meas)(1) = 1$.

We are now ready to prove \cref{thm:WassersteinDist_cpoCont}. Actually, we prove the following slight generalisation (recall that every probability measure on an analytic space is perfect).
\begin{theorem}
Let $\gamma\colon \measSp \to \giry\measSp \in \meas$ such that $\gamma(x)$ is a perfect measure (for each $x\in X$). Then the $\hat\sigma$ is \ocpo-continuous with respect to $\leq$.
\end{theorem}

\begin{proof}
Assume an $\leq$-increasing sequence $(d_i)_{i\in\nat}$ of pseudometrics over $\measSp$, then we need to show that the following equation for every $x,y\in X$.
\begin{equation}\label{eq:thm:Wasser_cpoCont}
  \inf_{\meas[c] \in K(\gamma_x,\gamma_y)}
            \int \sup_{i \in \nat} d_i \diff \meas[c]
        =
            \adjustlimits
            \sup_{i \in \nat}
            \inf_{\meas[c] \in K(\gamma_x,\gamma_y)}
            \int d_i \diff \meas[c]
    \text.
\end{equation}

\noindent
Let $x,y\in X$ and let $\meas_1 = \gamma_x $ and $\meas_2 = \gamma_{y}$.
    Recall the notion of charge which is the same thing as a measure but only finitely additive and note that
    \begin{align}
    \nonumber
    \MoveEqLeft
            \setBuilder
                { \meas[c] \text{ measure on } \sAlg \measTimes \sAlg }
                { \giry (\proj {1}) \meas[c] = \meas_1, \giry (\proj {2}) \meas[c] = \meas_2 }
    \\ \nonumber
        &=   \setBuilder*
                { \meas[c]  \text{ charge on } \sAlg \measTimes \sAlg  }
                { \giry (\proj {1}) \meas[c] = \meas_1, \giry (\proj {2}) \meas[c] = \meas_2 }
    \intertext{as $\meas_1,\meas_2$ are perfect \cite[D5]{RamachandranRüschendorf00} (see also \cite[Prop.\ 3]{Pachl79}, original result \cite[Thm.\ VIII]{RyllNardzewski53}). 
    Now applying the duality \eqref{eq:duality_charges_posFunctionals} and by suppress $\sAlg$ below in favour of readability  (i.e. $X\times X$ as a shorthand for the measurable space $\measSp \times \measSp$).
    }\nonumber
        &\stackrel{\operatorname{Int}}\cong
            \setBuilder*
                { c \in (\unifFct^* ( X\times X ) )  }
                { \begin{multlined}
                    c \text{ is positive}, c(1) = 1 \text{ and}  \\
                \text{for } i=1,2\colon
                (\unifFct^*\proj i) c = \operatorname{Int}(\meas_i)
                \end{multlined}}
    \\
    \label{eq:thm:WassersteinDist_cpoCont}
        &=  \begin{multlined}[t]
            \bigcap
                \setBuilder*
                    {\invImSet{\eval{f}}([0,\infty))}
                    { f\in \unifFct \bigl(X \times X \bigr), f \geq 0 }
            \cap \invImSet{\eval{1}}(\{1\}) \\
            {}\cap \bigcap\nolimits_{i=1,2}
                    \invImSet{(\unifFct^*\proj i)}(\{ \operatorname{Int} \meas_i \})
        \end{multlined}
    \intertext{where $\eval{f}\colon c \mapsto c(f)$ is the evaluation function.}
    \nonumber
        &\eqqcolon
            V_{xx'}
    \end{align}
    From the last representation it is apparent that $V_{xx'}$ is an intersection of weak-*-closed subsets of $ \unifFct^* ( X \times X ) $ as $
        \unifFct(\proj 1), \unifFct(\proj 2)\colon \unifFct X \to \unifFct ( X \times X)
    $ are continuous by functoriality.
    Hence $V_{xx'}$ is closed.
    Moreover the set $
        \setBuilder{\phi\in \unifFct^* ( X \times X ) }{ \phi(f) \in [-1,1] \text{ for any } f \text{ with } \Norm{f}{u} \leq 1 }
    $ is compact by Banach-Alaoglu theorem \cite[V.4.2]{DunfordSchwartz58}.
    This set also contains any positive $
        c \in \unifFct^* ( X \times X )
    $ with $ c(1) = 1 $
    (for any $f$ with $\Norm{u}{f} \leq 1$ note that $
            \abs{c (f)}
        \stackrel[{c \text{ linear}}]{}\leq
            \abs{c(\abs{f})}
        \stackrel[{c\text{ positive}}]{}=
            \int \abs{f} \diff c
        \hspace{-4em}
        \stackrel
            [\substack{c\text{ positive,} \\ \text{\cite[III.2.22, p.\ 119; III.1.5]
            {DunfordSchwartz58}} }]{}\leq
        \hspace{-4em}
            \int 1 \diff c
        =   1
    $).
    Thus $V_{xx'}$ is compact.

Set $U = [0,\infty) $ and define $\tilde d_x\colon U \to [0,1]$ by $
        \tilde d_x = d_i
    $ if $ x \in [i, i+1) $.
    Note that $\tilde d_x$ is an increasing function in $x$.
    Further define $ f\colon U \times V_{xx'} \to [0,1] \in \set$ by $
        f(x, c) \coloneqq c(\tilde d_x) = \int \tilde d_x \diff c
    $.
    For upper semi-continuity of $f$ in its first argument fix a $c \in V_{xx'}$ and take any $ r \in \mathbb R $.
    Observe that $
            \invImSet{f(\blank, c)} ( [0, r) )
        =   \setBuilder
                {x \in U }
                { \exists i \in \omega\colon x < i+1 \wedge c (d_i) < r }
    $ is open.
    As $r$ was chosen arbitrarily, upper-semicontinuity is proven.
    For lower semi-continuity of $f$ in its second argument fix an $x \in U$ and take again any $ r \in [0,\infty) $.
    Let $i$ be the element of $\omega$ with $ x \in [i, i+1) $.
    Observe that $
            \invImSet{f(x, \blank)} ( (r,\infty) )
        =   \invImSet{f(i, \blank)} ( (r,\infty) )
        =   \setBuilder{c \in V_{xx'}}{ c(d_i ) > r }
        =   \invImSet{\eval{d_i}} ( (r,\infty) )
    $ is open by definition of weak-*-topology.
    Since $x$ and $r$ were chosen arbitrarily, lower-semicontinuity is proven.
    For each $ c $ and $\lambda \in \mathbb R$ the level set $\setBuilder{x \in U}{ f(x, c) \geq \lambda }$ is of form  $[i, \infty) \subseteq \mathbb R$ and thus obviously convex.
    Convexity of $V_{xx'}$ is also quickly confirmed by noting that every operand in \cref{eq:thm:WassersteinDist_cpoCont} is closed under convex combination.
    For each $ x \in U $ and $\lambda \in \mathbb R$ the level set $
        \setBuilder
            {c \in \unifFct^* ( X \times X ) }
            { c(\tilde d_x) \leq \lambda}
    $ is convex by linearity of the $c$'s.
    As the intersection of convex sets is convex, the level set $
        \setBuilder
            {c \in V_{xx'} }
            { c(\tilde d_x) \leq \lambda}
    $ is convex.

    Then by \cref{lem:minimax_Sion},
    $
            \inf_{c\in V_{xx'}} \sup_{x \in [0,\infty)} c(\tilde d_x)
        =
            \sup_{x \in [0,\infty)} \inf_{c\in V_{xx'}} c(\tilde d_x)
    $.
    By definition of $\tilde d$ this becomes
    $
            \inf_{c\in V_{xx'}} \sup_{i \in \nat} c(d_i)
        =
            \sup_{i \in \nat} \inf_{c\in V_{xx'}} c(d_i)
    $.
    By the canonical isomorphism and
    applying Levi's theorem (monotone convergence theorem) \cite[123A]{Fremlin} the claim in \eqref{eq:thm:Wasser_cpoCont} follows. 
\end{proof}

\begin{arXivVersion}

\subsection{Proof of \cref{lem:dliftingForMP}}

Here one should point out that the triangle inequality actually requires perfectness (which is automatic as we are working over $\ana)$.

\begin{proof}
   The mapping $\sigma^c$ is well defined as predicate lifting is a map to $\spredicates$ (cf.~\cref{thm:predLifting_semimeas}), direct images are defines thereon (cf.~\cref{lem:directImage}) and convex combination is a measurable function.
   By unfolding the definition of $\sigma^c$ we obtain
   \[
            \dliftMP_X(d)((\meas,r),(\meas*,s))
        =   c\left(\inf_{\meas[c]\in K(\meas,\meas*)} \int d \diff \meas[c]\right) + (1-c)|r-s|
    \]

   It remains to verify that for any pseudo-distance $\Dist*$ the resulting in a pseudo-distance.
   Both summands $\inf_{\meas[c]\in K(\meas,\meas*)} \int d \diff \meas[c]$ and $|r-s|$ constitute distances in their own right as is known for the latter and for there former follows immediately from its dual characterisation by \cref{lem:KantorovicRubinstein}.
   The convex combination of pseudo-distances is again a pseudo-distance.
\end{proof}

\subsection{Proof of \cref{lem:dliftingForMDP}}

\begin{proof}
    The explicit formula is again just unfolding.
    From this formulas we see, that we are just dealing with a countable suprema of pseudo-distances.
    Measurability is preserved by countable suprema and pseudo-metrics even by arbitrary suprema.
\end{proof}

\end{arXivVersion}

\section{Topology}
    \label{sec:topology}
Recall that a topological space is \definiendum{$\mathrm R_2$} if any pair of topologically distinct point (i.e.\nolinebreak[3]\ $ 
    \exists U \in \topy \colon U \cap \{x,y\} \in \{\{x\}, \{y\} \}
$) are separated by disjoint open sets (i.e.\nolinebreak[3]\ $
    \exists U, V \in \topy\colon x \in U \wedge y \in V \wedge U \cap V = \emptyset 
$).
Also recall that a topological space is called $ \mathrm T_2 $ if it is Hausdorff and $\mathrm T_0$ if any pair of distinct point is separated by an open set.
Obviously, a $ \mathrm T_2 $-space is precisely an $\mathrm R_2$-space that is $\mathrm T_0$.
Any topological space can be transformed canonically into a $ \mathrm T_0 $ space by identifying all point that are not topologically distinct point, resulting in the so-called \definiendum{Kolmogorov quotient} $\operatorname{Kol}$.
We now generalise a bit the well-known Stone-Weierstraß theorem using the fact that $\mathrm T_0$-spaces form a reflective subcategory of topological spaces by the Kolmogorov quotient construction.
\begin{CALCOversion}
    Proofs of the following theorems are found in \cite{arXiv}.
\end{CALCOversion}

\begin{lemma}{Stone-Weierstraß}  
        \label{lem:StoneWeierstraß}
    Let $\topSp $ a compact $\mathrm R_2$-space.
    Let $ L $ be a set of continuous functions $\spSet\to \mathbb R$ such that $
        \min\{f,g\}, \max\{f,g\} \in L
    $ for all $f,g \in L$.
    If some continuous function $f \colon \spSet\to \mathbb R$ can be approximated at each pair of points by functions in $L$,
    then $f$ itself can also be approximated by functions in $L$ with respect to the uniform norm $\Norm{\blank}u$.
\end{lemma} 

\begin{arXivVersion}
\begin{proof}
    Let $\mathrm C( \topSp, \mathbb R )$ denote the set of continuous function $\topSp \to \mathbb R$ with uniform norm.
    Let $ \eta_{\blank} $ denote the unit witnessing the fact, that $ \mathrm T_0 $ spaces form a reflective subcateogory of topological spaces \cite[4.17(5)]{AdamekHerrlichStrecker04}.
    Note that $ \eta{\mathbb R} $ is a homeomorphism as $\mathbb R$ is already $\mathrm T_0$.
    Thus there is a bijection of functions $ f\colon \spSet \to \mathbb R $ with functions $f^\flat\colon \operatorname{Kol}\topSp \to \mathbb R$.
    View $ \mathrm C( \blank, \mathbb R ) $ as a contreavariant functor from $\Cat{Top}$ to normed spaces and bounded linear maps.
    A way to phrase the penultimate sentence is to say that $ \mathrm C( \eta_{\spSet}, \mathbb R ) $ is bijective.
    Moreover this map has operator norm $ \Norm{ \mathrm C( \eta_{\spSet}, \mathbb R ) }{} = 1$.
    Elementarily, a value $f^\flat(\hat x)$ is given by any $f(x)$ with $ x \in \hat x $.
    
    The theorem is found in this form in \cite[Lemma A.7.2]{Ash72} with ``$\mathrm R_2 $'' replaced by ``$\mathrm T_2 $''.
    So see that ``$\mathrm T_2 $'' can be weekend to ``$\mathrm R_2 $'' take any $L$ and $f$ as in the assumption of the theorem together with some $\varepsilon > 0$.
    Obviously, $ L' \coloneqq [L]^\flat $ is again closed under $\min$ and $\max$ as can be checked immediately by the elementary characterisation of $ (\blank)^\flat $.
    Likewise, the family $L'$ also approximates $ f^\flat $ as each pair of points.
    Moreover $\operatorname{Kol}\topSp$ is compact as the image of the compact space $\topSp$ and $\mathrm T_2$.
    Thus we can approximate $f^\flat$ up to $\varepsilon$ by some $ g' \in L' $.
    But so $ g \comp \eta_{s\spSet} $ approximates $f$ up to $\varepsilon$ as $ 
            \Norm{f-g \comp \eta_{s\spSet}}{u}
        =   \Norm{f^\flat \comp \eta_{\spSet} -g \comp \eta_{s\spSet}}{u}
        \leq\Norm{f^\flat - g}{u}
            \cdot \Norm{ \mathrm C( \eta_{\spSet}, \mathbb R ) }{}
        \leq\Norm{f^\flat - g}{u}
    $.
\end{proof}
\end{arXivVersion}

\begin{lemma}
        \label{eq:Kolmogorov_unit_proper}
    The unit of the Kolmogorov construction $\eta_{\spSet}\colon \spSet \to \operatorname{Kol}\spSet$ is proper, i.e.\nolinebreak[3]\ preimages of compact sets are compact.
\end{lemma}

\begin{arXivVersion}
\begin{proof}
    Given some compact subset $ L \subseteq \operatorname{Kol}\spSet $ set $ K \coloneqq \invImSet{\eta_{\topSp}}(L) $.
    To check compactness of $K$, take any open cover $\paving U \subseteq \topy \spSet$ of $K$, i.e.\nolinebreak[3]\ $ \bigcup \paving U \supseteq K  $.
    Set $ \paving V \coloneqq \dirImSet{f}[\paving U] $.
    As $\operatorname{Kol}\spSet$ has the quotient topology of $\eta_{\spSet}$, $ \paving V $ consists of open sets.
    By compactness of $L$ and $\bigcup \paving V = L $ we find a finite subpaving $\paving G \subseteq \paving V$ with $\bigcup \paving G = K$.
    Set $\paving F \coloneqq \invImSet{f}[\paving G]$.
    Obviously, $K \subseteq \bigcup \paving F$.
    We have to check that $ \paving F \setminus\{\emptyset\} \subseteq \paving U $.
    Assume that there is some nonempty $ F \in \paving F \setminus \paving U $.
    Let $U \in \paving U$ be such that $ F = \invImSet{f}(\dirImSet f (U)) $.
    Then also $U$ is nonempty.
    Also there is some $x \in F \setminus U $ as $ U \subseteq \invImSet{f}(\dirImSet f (U)) $.
    But this is in contradiction to nonemptyness of $F, U$, openness of $U$ and the construction of the Kolmogorov quotient.
\end{proof}
\end{arXivVersion}

\section{Proofs of lemmas and theorems from \cref{ssec:shapes}}

\subsection{Proof of \cref{lem:logicalDist_measurable}} 
\begin{proof}
    Let $\topy_\Sprache$ denote the topology given by topologisability of $\quantTheory$ and
    note that we have the following factorization of $\ldistance$
    \begin{multline*}
            \spSet\times\spSet
        \xrightarrow{\quantTheory\times\quantTheory}
                \Omega^{(\SpracheFormulas, \topy_{\Sprache})}
            \times \Omega^{(\SpracheFormulas, \topy_{\Sprache})}
        \cong
            (\Omega\times\Omega)^{(\SpracheFormulas, \topy_{\Sprache})}
    \\  \xrightarrow{\abs{ (\blank) - (\blank) }^{(\SpracheFormulas, \topy_{\Sprache})}}
            \Omega^{(\SpracheFormulas, \topy_{\Sprache})}
        \xrightarrow{( (\SpracheFormulas, \powerSet(\SpracheFormulas)) \to (\SpracheFormulas, \topy_{\Sprache}) )^*}
            \Omega^{(\SpracheFormulas, \powerSet(\SpracheFormulas))}
        \subseteq
            \textstyle\prod_{\SpracheFormulas} \Omega
    \\  \xrightarrow{ \sup_{\varphi \in \SpracheFormulas} \proj\varphi \colon f \mapsto \sup_{\varphi \in \SpracheFormulas} f(\varphi) }
            \Omega
    \text,
    \end{multline*}
    where the first map is measurable by assumption,
        the other maps with exception of the last one are continuous \cite[3.4.5, 3.4(1), 3.4(2), 2.6.3]{Engelking89} and
        the last one is lower semicontinuous \cite[1.7.14(a)]{Engelking89}.
    This implies the claim as the composition of measurable, continuous and semi-continuous functions is measurable.
\end{proof}

\allowdisplaybreaks
\subsection{Proof of \cref{thm:adequacy}}
\begin{proof}
    Set $\interpret{\blank} = \interpret{\blank}_\gamma$.
    We prove $ \bdistance \geq \interpret{\varphi}^\bullet \Dist[E] $ for each $\varphi\in\SpracheFormulas$ by structural induction over $\varphi$, where $\interpret{\varphi}^\bullet = \invImSet{(\interpret{\varphi} \times\interpret{\varphi})}$.
    Recall again that all logical symbols are also function symbols,
    so the proof consists of two cases: Function symbols and modal operators.

    Take a formula $\varphi = f(\varphi_1, \ldots, \varphi_n)$ for an arbitrary $n$-ary function symbol $f$---interpreted by a function also denoted by $f\colon [0,1]^n \mapsto [0,1]$---and formulas $\varphi_i$ with $i=1,\ldots,n$ with $ \bdistance \geq \interpret{\varphi_i}^\bullet \Dist[E] $ for each $i$.
    Given two states $x,y \in \spSet$ we find
    \begin{align*}
            \interpret{\varphi}^\bullet \Dist[E](x,y) \qquad\qquad\quad
        &\refrel{ssec:interpretation}=
            \abs{
                f(\interpret{\varphi_1}(x), \ldots, \interpret{\varphi_n}(x))
                -f(\interpret{\varphi_1}(y), \ldots, \interpret{\varphi_n}(y))
            }
    \\  &\textrel{f \text{ nonexpansive}}*\leq
            \Dist[\infty]*\bigl(
                (\interpret{\varphi_1}(x), \ldots, \interpret{\varphi_n}(x)),
                \interpret{\varphi_1}(y), \ldots, \interpret{\varphi_n}(y)
            \bigr)
    \\  &\textrel{\forall i\colon \bdistance \geq \interpret{\varphi_i}^\bullet \Dist[E]\text,}*[\text{definition of } \Dist[\infty]*]*\leq
            \bdistance(x,y)
    \text.
    \end{align*}
    As $x,y$ had been chosen arbitrarily, it follows that $ \bdistance \geq \interpret{\varphi}^\bullet \Dist[E] $.

    Turning in the final step to modal operators, take any $a \in \actions$ and assume that for a formula $\psi$ we already know that $ \bdistance \geq \interpret{\psi}^\bullet \Dist[E] $.
    Observe for any two states $x,y\in \spSet$
    \begin{align*}
            \qquad\interpret{\diamond_i \psi}^\bullet \Dist[E](x,y) \quad
        &\refrel{eq:modalOp_explicitly}=
            \abs*{
                \int \interpret{\psi} \diff \meas_{a,x} +_c \mathit{rew}_a(x)
                -\left(\int \interpret{\psi} \diff \meas_{a,y} +_c \mathit{rew}_a(y)\right)
            }
    \\  &\textrel{\text{definition of } +_c\text,}*[linearity of integral]=
            \abs*{\left(
                    \int \interpret{\psi} \diff \meas_{a,x}
                -   \int \interpret{\psi} \diff \meas_{a,y}
            \right) +_c \left(
                \mathit{rew}_a(x) - \mathit{rew}_a(y)
            \right)
            }
    \\  &\textrel{c, 1-c  \geq 0 \text{; definition of }+_c \text;}*[triangle inequality]\leq
            \left|
                    \int \interpret{\psi} \diff \meas_{a,x}
                -   \int \interpret{\psi} \diff \meas_{a,y}
            \right| +_c \left|
                \mathit{rew}_a(x) - \mathit{rew}_a(y)
            \right|
    \\  &=
            \left|
                    \int \interpret{\psi} \diff \meas_{a,x}
                -   \int \interpret{\psi} \diff \meas_{a,y}
            \right| +_c \left|
                \mathit{rew}_a^\bullet (\Dist[E])(x,y)
            \right|
    \intertext{let $\meas[c]$ always range over all coupling $ \couplings(\meas_{a,x}, \meas_{a,y}) $}
        &\textrel{definition of coupling}[\cref{def:pmetric_coupling}]=
            \left|
                \inf_{\meas[c] }
                    \int \interpret{\psi}(x) - \interpret{\psi}(y)
                    \meas[c](\diff x', \diff y')
            \right| +_c
                \mathit{rew}_a^\bullet (\Dist[E])(x,y)
    \\  &\textrel{Jensen's inequality}\leq
            \inf_{\meas[c]}
                \int \left|\interpret{\psi}(x) - \interpret{\psi}(y)\right|
                \meas[c](\diff x', \diff y')
             +_c
                \mathit{rew}_a^\bullet (\Dist[E])(x,y)
    \\  &=
            \inf_{\meas[c]}
                \int \interpret{\psi}^\bullet \Dist[E] \diff \meas[c]
             +_c
                \mathit{rew}_a^\bullet (\Dist[E])(x,y)
    \\  &\leq
                \adjustlimits
                \sup_{a \in \actions}
                \inf_{\meas[c]}
                    \int \interpret{\psi}^\bullet \Dist[E] \diff \meas[c]
             +_c
                \mathit{rew}_a^\bullet (\Dist[E])(x,y)
    \\  &\textrel*{\text{definition of }\sigma^B}*=
            \bigl(\gamma_B \circ \sigma^B( \interpret{\psi}^\bullet \Dist[E] )\bigr)(x,y)
    \\  &\textrel*{monotonicity from}
                [\text{\cref{thm:WassersteinDist_cpoCont}; }\bdistance \geq \interpret{\psi}^\bullet \Dist[E]]*\leq
            \bigl(\gamma_B \circ \sigma^B( \bdistance )\bigr)(x,y)
    \\  &=
            \bdistance(x,y)
    \text.
    \end{align*}
    As $x,y$ had been chosen arbitrarily, it follows that $ \bdistance \geq \interpret{\diamond_a \psi}^\bullet \Dist[E] $.
\end{proof}

Note that the usage of Jensen's inequality could be avoided using \cref{eq:dist_Sprache_def}.
But we chose to conduct the proof this way in view of possible further research into logics without negation.

\section{Proofs of lemmas and theorems from \cref{ssec:expressivity}}

\subsection{Proof of \cref{lem:approxAtPairOfPoints}} 
\begin{proof}
    Let $\varphi$ witness $f\in \interpret\Sprache$, i.e.\nolinebreak[3]\ $ \interpret\varphi = f $.
    Set $
        g = \interpret{((\varphi - fy)\wedge((hx-hy)\top))+ hy}
    $ which is defined due to $ 0 \leq hx - hy $.
    Note $
            \interpret{((\varphi - fy)\wedge((hx-hy)\top))+ hy}z
        =   \min\{1, ((fz - fy) \wedge (hx - hy) ) + hy \}
    $.
    Finally, check that
    \begin{align*}
            hx - gx
        &=  hx - \min\{1, ((fx - fy) \wedge (hx - hy) ) + hy \}
    \\  &\subseteq
            hx - (hx - \varepsilon , hx]
    \\  &=   [0, \varepsilon )
    \end{align*}
    and
    \[
            hy - gy
        =   hy - \min\{1, (0 \wedge (hx - hy) ) + hy \}
        =   0
    \text.
    \]
    Thus \cref{eq:lem:approxAtPairOfPoints_approxExact} follows.
    Finally, $ \cref{eq:lem:approxAtPairOfPoints_approxExact} \implies \cref{eq:lem:approxAtPairOfPoints_approx} $ follows from definition of approximation of a pair of points.
\end{proof}

\subsection{Proof of \cref{lem:KantorovicRubinstein}} 
Define the \definiendum{Kantorovic distance} for any probability measures $\meas, \meas* $ on $\measSp$
\begin{equation}
        \label{eq:KantorovicDist}
        K_1(\Dist*)(\meas, \meas*)
    =   \sup_{h \in \mathcal{L}_1(\spSet, \sAlg, \meas, \meas*, \Dist*)}
            \int h \diff (\meas - \meas*)
\text.
\end{equation}

\begin{proof}
    From \cite[(D*)]{RamachandranRüschendorf95} we have
    (using symmetry, triangle inequality and integrability of $\Dist*(\blank, x_0)$)
    that
    \begin{multline*}
            W_1(\Dist*)(\meas, \meas*)
        \coloneqq
            \inf_{\couples{\meas[c]}{\meas}{\meas*}} \int \Dist* \diff \meas[c]
        \\=
            \sup\setBuilder*
                {\int f \diff \meas + \int g \diff \meas*}
                {\begin{multlined}
                    \text{for } f \in \mathcal{L}^1\measdSp, g \in \mathcal{L}^1(\spSet, \sAlg, \meas*) \text{ with} \\
                        \forall x, y\in \spSet\colon f(x) + g(y) \leq \Dist*(x,y)
                \end{multlined}}
        \\\eqqcolon
            m(\Dist*)(\meas, \meas*)
    \end{multline*}
    As $ h(x) - h(y) \leq \Dist*(x,y) $ implies $ h(x) + (-h)(y) \leq \Dist*(x,y) $ we have certainly $
        K_1(\Dist*)(\meas, \meas*) \leq m(\Dist*)(\meas, \meas*)
    $.
    For the reverse direction take any $ f \in \mathcal{L}^1\measdSp $ and $ g \in \mathcal{L}^1(\spSet, \sAlg, \meas*)$ with $\forall x, y\in \spSet\colon f(x) + g(y) \leq \Dist*(x,y)$.
    Set
    \[
            h
        \coloneqq
            \lambdCalc x \in \spSet. \inf_{y\in \spSet} \Dist*(x,y) - g(y)
    \text.
    \]
    Then $f \leq h \leq -g $ and
    \begin{multline*}
            h(x) - h(x')
        =   \inf_{y\in \spSet} \Dist*(x,y) - g(y)
            + \sup_{y\in \spSet} g(y) - \Dist*(x',y)
    \\  \leq
            \sup_{y\in \spSet} \Dist*(x,y) - \Dist*(x',y)
        \stackrel{\text{triangle inequality}}\leq
            \Dist^*(x,x')
    \text.
    \end{multline*}
    So especially, $h$ is continuous with respect to the topology $\topy_{\Dist*}$.
    In case $ \topy_{\Dist*} \subseteq \sAlg $, it follows directly that $h$ is measurable.
    Otherwise note that $ g $ is nothing but $ \exists_{\proj 1} (\lambdCalc x\, y. d(x,y) - g(y)) $, and use that $\measSp$ is analytic (or smooth) to conclude measurability of $h$.
    Moreover $
            \int f \diff \meas + \int g \diff \meas*
        \leq\int h \diff \meas -h \meas*
    $ by $f \leq h \leq -g $.
    Hence $
        m(\Dist*)(\meas, \meas*) \leq K_1(\Dist*)(\meas, \meas*)
    $.
\end{proof}

\subsection{Proof of \cref{thm:interpretation_hemicompact}}
\begin{proof}
We use structural induction on $\psi \in \Shape(\Sprache)$ to assign a topology $\topy_{\psi}$ on the set $\widehat{\psi}$.
    \begin{itemize}
        \item Let $\psi = \top$ (base case). Then we let $\topy_\top$ be the singleton space.
        \item Let $\psi = \diamond_a \psi'$, for some $\psi'\in \Shape(\Sprache)$ and $a\in\actions$. Then we have $\widehat{\psi} \cong  \widehat {\psi'}$  canonically in $\set$. Thus, we let $\topy_{\widehat{\psi}} = \topy_{\widehat{\psi'}}$ since the latter exists by induction hypothesis.
        \item Let $\psi = f_i(\psi_1,\cdots,\psi_{n_i})$ for some $n_i$-ary function symbols $f_i$ and $\psi_j\in\Shape(\Sprache)$. Then we have the following bijection:
        \[
        \widehat \psi \ \cong\  Y_i \times \widehat{\psi_1} \cdots \widehat{\psi_{n_i}} \quad \text{in $\set$.}
        \]
        Using this bijection, we let $\topy_{\psi}$ is the product topology of $\topy_{Y_i}$ and $\topy_{\psi_j}$ (for $1\leq j \leq n_i$).
    \end{itemize}
    Finally, topologise $\SpracheFormulas$ as a coproduct $
        \coprod_{\psi \in \Shape(\Sprache)}
            (\widehat{\psi}, \topy_{\psi} )
    $ in the category $\Top$ of topological spaces.
    As $\Shape(\Sprache)$ is countable, $\SpracheFormulas$ is again second countable.
    In case of locally compact spaces the product is even locally compact again regardless of the size of the index set.


    In the next step use the fact \cite[3.4.4]{Engelking89}:
    \[
            \Omega^{\SpracheFormulas}
        =   \Omega^{\bigsqcup_{\psi \in \Shape(\Sprache)}
            (\widehat{\psi}, \topy_{\psi} )}
        \cong
            \prod_{\psi \in \Shape(\Sprache)}
                \Omega^{(\widehat{\psi}, \topy_{\psi} )}
    \text.
    \]
    By \cref{lem:FctSp_2ndCount} each
    $\Omega^{\widehat{\psi}}$ 
    is second countable.
    Thus so is the product $
        \prod_{\psi \in \Shape(\Sprache)} \Omega^{\widehat{\psi}}$. 
    As $\SpracheFormulas$ is second countable, $\Omega^{\SpracheFormulas}$ is hereditarily Lindel\"of \cite[3.8.D]{Engelking89}.
    Thus by \cite[4A3Dc(ii)]{Fremlin} to check measurability of $\quantTheory$ amounts to showing the measurability of each composition
    \begin{equation*}
         \spSet \xrightarrow{\quantTheory} \Omega^{\SpracheFormulas} \xrightarrow{\Omega^{\iota}} \Omega^{\widehat{\psi}}
    \text,\qquad
        \text{for each $\psi \in \Shape(\Sprache)$ and inclusion $\iota \colon \widehat{\psi} \hookrightarrow \Sprache$.}
    \end{equation*}
    Let $
            \quantTheory|_{\widehat{\psi}}
        \coloneqq
            \Omega^\iota \circ \quantTheory
    $ for each $ \psi \in \Shape(\Sprache) $.

    We prove the above composition is measurable by structural induction of $\psi \in \Shape(\Sprache)$.
    The base case is trivial as $\interpret{\top} \equiv 1$ is constant.
    For the induction step we distinguish the following cases:
    \proofsubparagraph{function symbols (including $\neg, \wedge, \vee$)}
    function symbols (including $\neg, \wedge, \vee$):
    Consider a formula $ \psi = f(t_1, \ldots, t_n) $ for an $n$-ary function symbol $f$ index by $Y$.
    Let $ \proj i $ denote the canonical projection from $
            \widehat{\psi}
        =   \widehat{t}_1 \times \ldots \times \widehat{t}_n \times Y
    $ to $
            \widehat{t_i}
    $ for each $i = 1,\ldots, n $. Further let $\quantTheory_n \colon X \to \prod_{1\leq i \leq n} \Omega^{\widehat{t}_i}$ (we abbreviate the subscripts to $\prod_{i\in n}$) be the diagonal map. This map is measurable since each $\quantTheory|_{\widehat{t}_i}$ (for each $i\in n$) is measurable by inductive hypothesis.


    We have that $\quantTheory|_{\widehat{\varphi}}$ can be expressed as a composition of a measurable with two continuous maps
    \begin{multline*}
        \spSet
            \xrightarrow{
                        \quantTheory_n
                }
                \Omega^{\widehat{t_1}} \times \ldots \times \Omega^{\widehat{t_n}}
             \xrightarrow{
                    (\proj 1)^* \times\ldots\times (\proj n)^*
                }
                        \Omega^{\widehat{t_1} \times\ldots\times \widehat{t_n} \times Y}
                    \times \ldots \times
                        \Omega^{\widehat{t_1} \times\ldots\times \widehat{t_n} \times Y}
            \\ \to
            (\Omega^{{} \times n})^{\widehat{t_1} \times\ldots\times \widehat{t_n} \times Y}
            \xrightarrow{ f^{\widehat{t_1} \times\ldots\times \widehat{t_n} \times Y} }
            \Omega^{\widehat{t_1} \times\ldots\times \widehat{t_n} \times Y}
    \end{multline*}
    by standard result for the compact-open topology \cite[3.4(1), 3.4.5, 3.4(2)]{Engelking89}.
    Thus this composition is measurable.
    \proofsubparagraph{Case $\psi' = \diamond_a \psi$, for some $\psi \in \Shape (\Sprache)$.} By induction hypothesis, $\quantTheory(\blank)|_{\widehat{\psi}}$ is measurable and we must show that $\quantTheory(\blank)|_{\diamond_a \psi}$ is measurable as a function to $\Omega^{\topSp*} $, where $\topSp* = \widehat{\psi'} = \widehat{\psi}$.

    As convex combination $\blank +_c \blank$ is measurable, it suffices to show that both $\psi_1 \coloneqq \int \interpret{\psi} \diff \gamma_{x,a} $ and $ \psi_2(x) \coloneqq \gamma^1_{a}(x)$ are continuous functions for each $x$ and measurable in $x$ considered as $\Omega^{\topSp*} $-valued functions.

    In case of $\psi_2 $ we use 
    the following composition (below $k\colon \widehat{\psi}_2 \to 1=\{\bullet\}$ given by the constant mapping $t\in\widehat{\psi}_2 \mapsto \bullet$)
    \[
           \spSet \xrightarrow{\interpret{\reward_a \psi}} \Omega
            \xrightarrow{\cong} \Omega^{1}
            \xrightarrow{\Omega^k}
                \Omega^{\widehat{\psi}_2}
    \]
    consisting again of a measurable map followed by continuous maps, cf.\nolinebreak[3]\ \cite[2.6.8, 4.2.17, 3.4(2)]{Engelking89}.
    Thus $ \quantTheory|_{\widehat{\psi}_2} $ is measurable.

    Finally, for the formula $\psi_1$ involving integration we have to check that
    \[
        \quantTheory|_{\widehat{\psi}_1}
       = \lambdCalc x. \lambdCalc t \in \widehat{\psi}_1.
            \int \interpret{t} \diff \gamma_{x,a}.
        \eqqcolon
            \lambdCalc x. \lambdCalc t \in \widehat{\psi}_1.
            \mathrm E_a t
    \]
    is valued in $\Omega^{\topSp*} $ and measurable.
    \begin{claim}
        $\lambdCalc x. \lambdCalc t \in \widehat{\psi}_1.
            \mathrm E_a t$ is valued in $\Omega^{\topSp*} $.
    \end{claim}
    \begin{claimproof}
        We observe that $ \quantTheory|_{\widehat{\psi}_1}(x) $ is actually continuous for each $x \in X$ by a sequence argument (which suffices due to second countability \cite[1.6.14-15]{Engelking89}):
        Given a sequence $t_i \in \widehat{\mathrm E_a \psi}$ converging to some $t \in \widehat{\mathrm E_a \psi}$ we have $
                \quantTheory(x)(t_i)
            \xrightarrow{i\to\infty}
                \quantTheory(x)(t)
        $ for every $x \in \spSet$ by assumption (convergence with respect to compact-open topology implies convergence at each point, as points are compact).
        Thus by dominated convergence theorem \cite[123C]{Fremlin} we conclude
        \[
                \quantTheory(x)( \mathrm E_a t_i)
            \xrightarrow{i\to\infty}
                \quantTheory(x)( \mathrm E_a t)
        .\]    
    \end{claimproof}
    
    \begin{claim}
        $\lambdCalc x. \lambdCalc t \in \widehat{\psi}_1.
            \mathrm E_a t$ is measurable.
    \end{claim}
    \begin{claimproof}
        We first observe that being second countable $ Y $ is a k-space \cite[3.3.20]{Engelking89}.
        This is to say that there is some locally compact space $ Y' $ such that $Y$ is a quotient thereof.
        In this case $Y $ is the direct limit of it compact subspaces $ \paving K $ ordered by inclusion.
        In this case $\Omega^{\topSp*}$ is canonically homeomorphic to the inverse limit $ \varprojlim{}_{K\in \paving K } \Omega^K $ \cite[3.4.11]{Engelking89}.
        As the space $ \Omega^Y $ is also Lindelöf \cite[3.8.D]{Engelking89},
        is suffices to prove measurablity for any projection to some $ \Omega^K $ with $K \in \paving K $ \cite[4A3Db]{Fremlin}.
        This is to say that we can assume $\topSp*$ to be compact.
        For a compact space, the compact-open topology coincides with the uniform topology \cite[4.2.17]{Engelking89}.
        Thus it is metrisable.
        Recall that for metrisable spaces separability, second countability and Lindelöf are equiveridical \cite[4.1.16]{Engelking89} and hereditary properties.
        Thus by \cite[3.8.D]{Engelking89} $ \image \quantTheory|_{\widehat{\psi}} \subseteq \Omega^Y $ is second countable and separable.
        So it suffices to verify measurability of $ \quantTheory|_{\widehat{\psi}} $ on a countable subbase of $\Omega^Y$ (observe that any open can be formed by a countable union of elements of the base generated by the subbase).
        A countable subbase for the uniform topology of $ \image ( \quantTheory|_{\widehat{\psi}} ) \subseteq \Omega^Y $ is given by all open $\varepsilon$-balls around a sequence $
            \quantTheory|_{\widehat{\psi}}(x_i)
        $ dense in $\image \quantTheory|_{\widehat \psi}$ indexed by $i \in \omega$ and $\varepsilon \in 1/{\mathbb N}$ (this uses that $\Omega^Y$ is metrisable).
        Choose also a dense sequence $ y_j \in Y $.
        Set $ r_{i,j} \coloneqq \quantTheory|_{\widehat{\psi}}(x_i)(y_j) $.
        Observe
        \begin{align*}
        \MoveEqLeft[8]
                \invImSet{\quantTheory|_{\widehat{\psi}}}* \left(\varepsilon\text{-ball around } \quantTheory|_{\widehat{\psi}}(x_i)\right)
        \\  &\textrel{r_{i,j} = \quantTheory|_{\widehat{\psi}}(x_i)(y_j)}*=
                \invImSet{\quantTheory|_{\widehat{\psi}}}*\left(
                    \setBuilder
                        {f \in \Omega^Y}{ \forall y \in Y\colon \abs{r_{i,j} - f(y)} < \varepsilon }
                    \right)
        \\  &\textrel{triangle inequ.\nolinebreak[3]\ in $\Omega$}=
                    \invImSet{\quantTheory|_{\widehat{\psi}}}*\left(\setBuilder
                        {f \in \Omega^Y}
                        { \forall j \in \omega\colon \abs{r_{i,j} - f(y_j)} < \varepsilon }
                \right)
        \\  &=   \bigcap_{j \in \omega}
                    \invImSet{\interpret[\big]{(\mathrm E_a \hat\psi)_{y_j}}}*
                        \bigl( (r_{i,j} - \varepsilon, r_{i,j} + \varepsilon) \bigr)
        \intertext{as each $\interpret{(\mathrm E_a \hat\psi)_{y_j}}(\blank)$ is measurable}
            &\in
                \sAlg
        \text.
        \end{align*}
        As $\varepsilon$ and $i$ were arbitrary, $M$ is measurable.
    \end{claimproof}
\end{proof}

\end{document}